\documentclass[11pt]{article}
\usepackage{lineno,amsmath,hyperref}
\usepackage{epstopdf}
\usepackage{geometry}
\usepackage{amsthm}
\usepackage{color}
\usepackage{amsmath}
\usepackage{amssymb}
\usepackage{algorithm}
\usepackage{algorithmic}
\usepackage{algorithm,float}
\usepackage{amsmath}
\usepackage{listings}
\usepackage{extarrows}
\usepackage{xcolor}
\usepackage{graphicx}
\usepackage{makecell}
\newtheorem{theorem}{Theorem}[section]
\newtheorem{lemma}[theorem]{Lemma}

\newtheorem{property}[theorem]{Property}

\def\R{{\mathbb{R}}}

\def\N{{\mathbb{N}}}

\def\P{{\mathcal{P}}}
\def\C{{\mathcal{C}}}
\def\V{{\mathcal{V}}}

\def\M{{\mathcal{M}}}
\def\RC{{\mathcal{R}}}

\def\span{{\hbox{\rm{Span}}}}

\def\ord{{\rm{ord}}}
\def\d{{\hbox{\rm{d}}}}

\lstdefinelanguage{Maple}{
   keywords={if, while, do, else, end, for, from, to,then},
   keywordstyle=\color{blue}\bfseries,
   ndkeywords={class, export, boolean, throw, implements, import, this},
   ndkeywordstyle=\color{darkgray}\bfseries,
   identifierstyle=\color{black},
   sensitive=false,
   comment=[l]{//},
   morecomment=[s]{/*}{*/},
   commentstyle=\color{purple}\ttfamily,
   stringstyle=\color{red}\ttfamily,
   morestring=[b]',
   morestring=[b]"
}

\lstdefinelanguage{SOStools}{
   keywords={syms,sosprogram,monomials,sosineq,sossetobj,sossolve,sosgetsol,sospolyvar},
   keywordstyle=\color{blue}\bfseries,
   ndkeywords={syms,sosprogram,monomials,sosineq,sossetobj,sossolve,sosgetsol},
   ndkeywordstyle=\color{blue}\bfseries,
   identifierstyle=\color{black},
   sensitive=false,
   comment=[l]{//},
   morecomment=[s]{/*}{*/},
   commentstyle=\color{purple}\ttfamily,
   stringstyle=\color{red}\ttfamily,
   morestring=[b]',
   morestring=[b]"
}

\modulolinenumbers[5]


\begin{document}
\title{A Generalization of the Concavity of R\'{e}nyi Entropy Power}
%
%The Third and Fourth Derivatives of Multivariate Costa's Entropy Power Inequality}
\author{Laigang Guo, Chun-Ming Yuan, Xiao-Shan Gao\\
KLMM, Academy of Mathematics and Systems Science\\
 Chinese Academy of Sciences, Beijing 100190, China\\
University of Chinese Academy of Sciences, Beijing 100049, China}
%\\Email: xgao@mmrc.iss.ac.cn}
\date{\today}
\maketitle

\begin{abstract}
\noindent
Recently, Savar\'{e}-Toscani proved that the R\'{e}nyi entropy power of general probability densities solving the $p$-nonlinear heat equation in $\mathbb{R}^n$ is always a concave function of time, which extends Costa's concavity inequality for Shannon's entropy power to R\'{e}nyi entropies.
In this paper, we give a generalization of Savar\'{e}-Toscani's result
by giving a class of sufficient conditions of the parameters under which the concavity of the R\'{e}nyi entropy power is still valid. These conditions are quite general
and include the parameter range given by Savar\'{e}-Toscani as special cases.
Also, the conditions are obtained with a systematical approach.

\end{abstract}

%% The paper must be self-contained. However, if you are referring to
%% a full version for checking certain proofs, please provide the
%% publically accessible location below.  If the paper is completely
%% self-contained, you can remove the following line from your
%% submission.

%\textit{A full version of this paper is accessible at:}
%\url{https://2021.ieee-isit.org/}

\section{Introduction}

The {\em $p$-th R\'{e}nyi entropy} \cite{Cover1992,Gardner2002} of a probability density function $f: \mathbb{R}^n\rightarrow \mathbb{R}$ is defined as
\begin{eqnarray}
\label{I1}
H_p(f):=\dfrac{1}{1-p}{\rm log}\int_{\mathbb{R}^n}f^p(x)dx,
\end{eqnarray}
for $0<p<+\infty$, $p\neq1$.
The $p$-th R\'{e}nyi entropy  power is given by
\begin{eqnarray}\label{I2}
N_p(f):={\rm exp}(\frac{\mu}{n}H_p(f)),
\end{eqnarray}
where $\mu$ is a real valued parameter. The R\'{e}nyi entropy for $p=1$ is defined as the limit of $H_p(f)$ as $p\rightarrow1$. It follows from definition \eqref{I1} that
\begin{eqnarray*}
H_1(f)=\underset{p\rightarrow1}{\rm lim}H_p(f)=-\int_{\mathbb{R}^n}f(x){\rm log}f(x)dx,
\end{eqnarray*}
which is Shannon's entropy. In this case, the proposed R\'{e}nyi entropy power of index $p=1, \mu=2$, given by \eqref{I2}, coincides with Shannon's entropy power
\begin{eqnarray}
\label{I3}
N_1(f):={\rm exp}(\dfrac{2}{n}H_1(f)).
\end{eqnarray}

Shannon's {\em entropy power inequality (EPI)}
is one of the most important information inequalities~\cite{Shannon1948}, which has many proofs, generalizations, and applications~\cite{Stam1959,Blachman1965,Lieb1978,VerduGuo2006,Rioul2011,Bergmans1974,Zamir1993,Liu2007,Wang2013}.
In particular, Costa presented a stronger version of the EPI~\cite{Costa1985}.

%Let $X$ be an n-dimensional random vector with {\em probability density} $p(x)$.
%For $t>0$, define $X_t\triangleq X+Z_t$, where $Z_t\thicksim N_n(0,tI)$ is an independent standard Gaussian random vector with covariance matrix $t\times I$.
%
Let $X_t\triangleq X+N_n(0,tI)$ be the $n$-dimensional random vector
introduced by Costa~\cite{Costa1985,GYG2020,GYG2020M}
and $u(x_t)$ the {\em probability density} of  $X_t$,
which solves the heat equation in the whole space $\mathbb{R}^n$,
\begin{eqnarray}
\label{HeatEqu1}
\dfrac{\partial}{\partial t}u(x_t)=\Delta u(x_t).
\end{eqnarray}
{\em Costa's differential entropy} is defined to be
{\begin{equation}
\label{1.1}
H(u(x_t))=-\int_{\R^n} u(x_t)\log u(x_t)\d x_t.
\end{equation}}
Costa~\cite{Costa1985} proved that
the {\em Shannon entropy power}
%\begin{equation}\label{1.5}
$N(u)=\frac{1}{2\pi e}e^{(2/n)H(u)}$
%\end{equation}
is a concave function in $t$,
that is $(\d/\d t)N(u)\ge0$ and $(\d^2/\d^2 t)N(u)\le0$.
Several new proofs and generalizations for Costa's EPI were given~\cite{Dembo1989,Villani2000,Toscani2015}.

Savar\'{e}-Toscani~\cite{Savare2014} proved that the {\em concavity of entropy power} is a property which is not restricted to the Shannon entropy power \eqref{I3} in connection with the heat equation \eqref{HeatEqu1}, but it holds for the $p$-th R\'{e}nyi entropy power \eqref{I2}. They put it in connection with the solution to the {\em nonlinear heat equation}
\begin{eqnarray}
\label{HeatEqu2}
\dfrac{\partial}{\partial t}u(x_t)=\Delta u(x_t)^p
\end{eqnarray}
posed in the whole space $\mathbb{R}^n$ and $p\in\R_{>0}$.

In this paper, we give a generalization for {\em concavity of $p$-th R\'{e}nyi entropy power (CREP)}.
Precisely, we give a propositional logic formula $\Phi(n,p,\mu)$ such that
if $n\in\N,p,\mu\in\R$ satisfy this formula, then the CREP holds.
The condition $\Phi(n,p,\mu)$ extends the parameter's range of the CREP given by Savar\'{e}-Toscani~\cite{Savare2014} and contains much more cases.
%In fact, $\Phi(n,p,\mu)$ gives the necessary and sufficient condition for CREP in certain cases.

The formula $\Phi$ is obtained using a systematic procedure
which can be considered as a parametric version of that given  in \cite{GYG2020,GYG2020M,Zhang2018}, where parameters $n,p,\mu$ exist in the formulas.
The procedure reduces the proof of the CREP to check the
semi-positiveness of a quadratic form  whose coefficients are polynomials in the parameters $n,p,\mu$.
In principle, a necessary and sufficient condition for the parameters to satisfy this property can be computed with the quantifier elimination~\cite{QE}.
In this paper, the problem is special and an explicit proof is given.

The rest of this paper is organized as follows.
In Section 2, we give the proof procedure and to prove concavity of entropy powers in the parametric case.
In Section 3, we present the generalized version of  CREP  using the proof procedure.
In Section 4, conclusions are presented.

\section{Proof Procedure}
\label{sec-2}
In this section, we give a procedure to prove the { CREP}.
%Lemma \ref{lm-pr1} is also new which reduces the { CREP} into an integral inequality.
To make the paper concise, we only give those steps that are needed in this paper.

\subsection{Notations}

Let $x_t=[x_{1,t},x_{2,t},\ldots,x_{n,t}]$ be a set of variables depending on $t$ and
$$
\begin{aligned}
\d^{(i)} x_t=\d x_{1,t}\d x_{2,t}\ldots \d x_{i-1,t}\d x_{i+1,t}\ldots \d x_{n,t}, i=1,2\ldots,n.
\end{aligned}
$$
Let $[n]_0 = \{0,1,\ldots,n\}$ and $[n] = \{1,\ldots,n\}$.
To simplify the notations, we use $u$ to denote $u(x_t)$ in the rest of the paper.
Denote
{\footnotesize
$$\mathcal{P}_{n} =\{\frac{\partial^h u}{\partial^{h_1} x_{1,t}\cdots \partial^{h_n} x_{n,t}}:
h = \sum_{i=1}^n h_i, h_i\in \N\}$$}
to be the set of all derivatives of $u$ with respect to the differential operators
$\frac{\partial}{\partial x_{i,t}},i=1,\ldots,n$
and
$$\RC=\R[n,p,\mu][\mathcal{P}_{n}]$$
 to be the set of polynomials in $\mathcal{P}_{n}$ with coefficients in $\R[n,p,\mu]$, where $n,p,\mu$ are  parameters.
For $v\in\mathcal{P}_{h,n}$, let $\ord(v)$   be the order of $v$.
%and we use $u^{(h)}$ to denote an $h$th-order derivative of $u$  if no confusing is caused.
For a monomial $\prod_{i=1}^r v_i^{d_i}$ with $v_i\in {\mathcal{P}}_{n}$,
its {\em degree}, {\em order}, and {\em total order}
are defined to be $\sum_{i=1}^r d_i$, $\max_{i=1}^r \ord(v_i)$,
 and  $\sum_{i=1}^r d_i\cdot \ord(v_i)$, respectively.
%
%Let $P\in\R[\mathcal{P}]$. The {\em degree} and {\em order} of $P$ are respectively the maximal degree and order of its monomials.

A polynomial in $\RC$ is called a $k$th-order
{\em differentially homogenous polynomial} or simply
a  $k$th-order {\em differential form},
if all its monomials have degree $k$ and total order $k$.
Let $\M_{k,n}$ be the set of all monomials which have degree $k$ and total order $k$.
Then the set of $k$th-order differential forms is
an $\R$-linear vector space generated by $\M_{k,n}$, which is denoted as $\span_\R(\M_{k,n})$.
We will use Gaussian elimination in $\span_\R(\M_{k,n})$ by treating the monomials as variables.
We always use the {\em lexicographic order for the monomials}  defined   in \cite{GYG2020,GYG2020M}.

\subsection{The proof procedure}
\label{sec-p2}
In this section, we give the  procedure to prove the {\em CREP}.
The property $\frac{\d}{\d t}N_p(u)\geq0$ can be easily proved~\cite{Savare2014}.
We focus on proving $\frac{\d^2}{\d^2 t}N_p(u)\le0$.
The procedure consists of four steps.

In step 1, we reduce the proof of {\em CREP} into the proof of an integral inequality,
as shown by the following lemma whose proof is given in section \ref{sec-p4}.
%We first compute $\frac{\d^2}{\d^2 t}N_p(u)$. For convenience, introduce the notation
%$u_{i,j}= \frac{\partial^{i+j} u}{\partial^{i} x_{a,t}\partial^{j} x_{b,t}}$,
%where $a,b$ are parameters taking values in $[n]$.

\begin{lemma}
\label{lm-pr1}
Proof of $\frac{\d^2}{\d^2 t}N_p(u)\le0$ can be reduced to show
{%\footnotesize
\begin{equation}
\label{eq-tt1}
\begin{array}{ll}
\displaystyle{\int_{\R^n}u^{3p-6}E_{2,n}\d x_t} \ge0
\end{array}
\end{equation}}
under the condition $p\geq1-\frac{\mu}{n}$,
where $E_{2,n} =\sum_{a=1}^n\sum_{b=1}^n E_{2,n,a,b}$ is a $4$th-order differential form in $\R[n,p,\mu][{\P}_{2,n}]$
and
{%\scriptsize
\begin{equation}
\label{eq-pm}
{\P}_{2,n} =\{
\frac{\partial^h u}{\partial^{h_1} x_{a,t}\partial^{h_{2}} x_{b,t}}:
%h=\sum_{i=1}^m h_i\in[2m]_0;
h\in[3]_0; a,b\in[n]\}.
\end{equation}}
\end{lemma}

Then the problem {$\frac{\d^2}{\d^2t}N_p(u)\leq0$} can be transformed to {  $\int u^{3p-6}E_{2,n}\d x_t\geq0$}.
Thus, Lemma \ref{lm-pr1} is proved.

In step 2, we compute the constraints which are relations satisfied by the probability density $u$ of $X_t$.
%which are given in \eqref{eq-2cons}.
Since $E_{2,n}$ in \eqref{eq-tt1} is a $4$th-order differential form, we need only
the constraints which are $4$th-order differential forms.
A $4$th-order differential form $R$ is called an {\em equational or inequality  constraint} if
\begin{equation}
\label{eq-c1}
{\int_{\R^n}u^{3p-6}R\d x_t} =0 \hbox{ or } {\int_{\R^n}u^{3p-6}R\d x_t} \ge0.
\end{equation}
The method to compute the constraints is given in section \ref{sec-2oc}.
Suppose that the equational and inequality constraints are respectively
\begin{equation}
\label{eq-cons}
\C_{E} =\{R_i,\,|\, i=1,\ldots,N_1\}\hbox{ and }
\C_{I} =\{I_i,\,|\, i=1,\ldots,N_2\}.
\end{equation}

In step 3, we find a propositional formula $\Phi(n,p,\mu)$ such that when $n,p,\mu\in\R$ satisfy $\Phi$,
\begin{eqnarray}
\label{eq-S1}
\exists c_j, e_i\in\R, s.t. \,\,  E_{2,n}-\sum_{i=1}^{N_1} e_i R_i -\sum_{j=1}^{N_2} c_j I_j =S\ge0
\hbox{ and } c_j\ge0,j=1,\ldots,n_2.
\end{eqnarray}
Details of this step and the formula $\Phi(n,p,\mu)$ are given in section \ref{sec-repi}.

To summarize the proof procedure, we have
\begin{theorem}
\label{th-m1}
The {\em CREP} is true if  $\Phi(n,p,\mu)$ is valid.
\end{theorem}
\begin{proof}
By Lemma \ref{lm-pr1}, we have the following proof:
{
\begin{equation}
\label{eq-proof}
\begin{aligned}
&\int_\R u^{3p-6}{{E}_{2,n}}\d x_t\\
\overset{\eqref{eq-S1}}{=}&\int_\R u^{3p-6}({\sum_{i=1}^{N_1} e_i R_i+\sum_{j=1}^{N_2} c_j I_j+ S})\d x_t\\[0.1cm]
\overset{S1}{=}&\int_\R u^{3p-6}(\sum_{j=1}^{N_2} c_j I_j+ S)\d x_t\\[0.1cm]
\overset{S2}{\geq} &\int_\R u^{3p-6}{S}\d x_t
\overset{S3}{\geq}0.
\end{aligned}
\end{equation}}
Equality S1 is true, because $R_i$ is a  equational constraint.
Inequality S2 is true, because $I_j$ is an inequality constraint.
Inequality S3 is true, because $S\ge0$ under the condition $\Phi(n,p,\mu)$.
\end{proof}

\subsection{The equational constraints}
\label{sec-2oc}

In this section, we show how to find the second order  equational constraints.
%
%\begin{define}\label{def-41}
%
A {\em second order equational constraint} is a $4$th-order differential form in $\R[n,p,\mu][\P_{2,n}]$
such that $\int_{\R^n}\ u^{3p-6}{R}\ \d x_t=0$.

We introduce the following notations
{\footnotesize
\begin{equation}
\label{eq-v2}
{\mathcal V}_{a,b} =\{ \frac{\partial^h u}{\partial^{h_1} x_{a,t} \partial^{h_2} x_{b,t}}:
h=h_1+h_2\in[3]_0\},
\end{equation}}
where $a,b$ are variables taking values in $[n]$.
Then $\P_{2,n}=\cup_{a=1}^n\cup_{b=1}^n \V_{a,b}$.

We need the following property.
\begin{property}\label{lemma4}
Let $a,r,m_i,k_i \in \N_{>0}$ and $u^{(m_i)}$ an $m_i$th-order derivative of $u$.
If $u(x_t)$ is a smooth, strictly positive and rapidly decaying probability density, then
{\footnotesize
\begin{equation}\label{lem4}
\begin{aligned}
\displaystyle{\int_{-\infty}^{\infty}\ldots\int_{-\infty}^{\infty}u^{3p-2}\left[\prod\limits_{i=1}^{r}
\frac{[u^{(m_i)}]^{k_i}}{u^{k_i}}\right]\Bigg|_{x_{a,t}=-\infty}^{\infty}\d^{(a)}} x_t=0,
\end{aligned}
\end{equation}}
with {\small $\sum_{i=1}^{r}k_im_i=4,\ \sum_{i=1}^{r}k_i=4$}.
\end{property}
When $p\geq2$, Property \ref{lemma4} follows from  \cite{Treves2016}.
While $0<p<2, p\neq1$, we make the assumption that  $u(x_t)$ also satisfies Property \ref{lemma4}.

Using Property \ref{lemma4}, we can compute the second order equational constraints
using the method given in \cite{GYG2020,GYG2020M}:
%
%\begin{lemma}
%\label{lm-CS1}
%Integrating by parts, we can compute the $2$th-order integral constraints
%$\C_{2,n} = \{R_i,i=1,\ldots,N_1\}$.
%\end{lemma}
%
%Based on Property \ref{lemma4}, we can obtain the second order integral constraints.
{%\small
\begin{equation}
\label{eq-2cons}
\C_{2,n}=\{R_{i,a,b}\,:\,i=1,\ldots,28\}\subset\R[n,p,\mu][{\mathcal V}_{a,b}],
\end{equation}}
where $R_{i,a,b}$  can be found in the Appendix.
Note that $a,b$ are variables taking values in $[n]$.

We use an example to show how to obtain these constraints.
Starting from a monomial
$u\frac{\partial^2u}{\partial^2x_{a,t}}(\frac{\partial u}{\partial x_{a,t}})^2$
with degree 4 and total order 4,
by integral by parts, we have
\begin{equation}\begin{array}{ll}
\int u^{3p-6}u\frac{\partial^2u}{\partial^2x_{a,t}}(\frac{\partial u}{\partial x_{a,t}})^2\d x_t\\
=\int_{-\infty}^{\infty}\ldots\int_{-\infty}^{\infty} [u^{3p-5}\frac{\partial u}{\partial x_{a,t}}(\frac{\partial u}{\partial x_{a,t}})^2]\Big|_{x_{a,t}=-\infty}^{\infty}\d^{(a)} x_t\\
-\int \frac{\partial u}{\partial x_{a,t}}[\frac{\partial}{\partial x_{a,t}}(u^{3p-5}(\frac{\partial u}{\partial x_{a,t}})^2)] \d x_t\\
\overset{\eqref{lem4}}{=}-\int \frac{\partial u}{\partial x_{a,t}}[\frac{\partial}{\partial x_{a,t}}(u^{3p-5}(\frac{\partial u}{\partial x_{a,t}})^2)] \d x_t.
\end{array}
\end{equation}
Then,
{%\footnotesize
\begin{equation}\begin{aligned}
\int u^{3p-6}u\frac{\partial^2u}{\partial^2x_{a,t}}(\frac{\partial u}{\partial x_{a,t}})^2
+\frac{\partial u}{\partial x_{a,t}}[\frac{\partial}{\partial x_{a,t}}(u^{3p-5}(\frac{\partial u}{\partial x_{a,t}})^2)]\d x_t\\
=\int u^{3p-6}[3p(\frac{\partial u}{\partial x_{a,t}})^4+3u\frac{\partial^2u}{\partial^2x_{a,t}}(\frac{\partial u}{\partial x_{a,t}})^2-5(\frac{\partial u}{\partial x_{a,t}})^4] \d x_t=0.
\end{aligned}\end{equation}
}
We obtain a 2th-order constraint:
$R_{1,a,b}=3p(\frac{\partial u}{\partial x_{a,t}})^4+3u\frac{\partial^2u}{\partial^2x_{a,t}}(\frac{\partial u}{\partial x_{a,t}})^2-5(\frac{\partial u}{\partial x_{a,t}})^4$. The other 27 constraints in $\C_{2,n}$ are obtained in the same way.

\subsection{Proof of Lemma \ref{lm-pr1}}
\label{sec-p4}

%First, we gave properties of $H_p(u)$ given by Savar\'{e}-Toscani~\cite{Savare2014}.
%{\small
%\begin{equation}\begin{array}{ll}
% \frac{\d H_p(u)}{\d t}=\displaystyle{\dfrac{1}{\int_{\R^n} u^p\d x_t}\int_{\R^n}\frac{\|\nabla u^p\|^2}{u}\d x_t}
% =I_p(u),\label{H2}
%%
%%\mathbb{E}[\nabla^2\log u]&=&-\int_{\R^n}\frac{\|\nabla u\|^2}{u}\d x_t.\label{H3}
%\end{array}\end{equation}}
%where
%{\small
%$\nabla u =(\frac{\partial u}{\partial x_{1,t}},\ldots,\frac{\partial u}{\partial x_{n,t}})$},
%and $I_p(u)$ is called the $p$-weighted Fisher information.
%%
%%Equation \eqref{H1} shows that $u$ satisfies the {\em heat equation}.
%%and \eqref{H2} implies $D(1,n)$: $\frac{\d }{\d t} H(X_t)\ge0$.
%%
%Thus we have
%{\footnotesize
%\begin{equation}\label{dN1}
%\frac{\d}{\d t}N_p(u)=\frac{\mu}{n}\frac{\d H_p(u)}{\d t}{\rm e}^{\frac{\mu}{n}H_p(u)}=\frac{\mu}{n}I_p(u){\rm e}^{\frac{\mu}{n}H_p(u)}\geq0.
%\end{equation}}
%%
%The property $\frac{\d}{\d t}N_p(u)\geq0$ can be easily proved~\cite{Savare2014}. In this paper, we focus on  $\frac{\d^2}{\d^2 t}N_p(u)$.
%
%
%
We first prove several lemmas.
\begin{lemma}
{%\scriptsize
\begin{equation}\begin{array}{ll}
\frac{\d H_p(u)}{\d t}=\frac{p}{1-p}\frac{\int u^{p-1}\Delta u^p\d x_t}{\int u^p \d x_t},\label{dh1}\\
\end{array}\end{equation}
\begin{equation}
\frac{\d^2H_p(u)}{\d^2t}=\frac{p}{1-p}\frac{\frac{\partial}{\partial t}(\int u^{p-1}\frac{\partial u}{\partial t}\d x_t)\int u^p\d x_t-p(\int u^{p-1}\frac{\partial u}{\partial t}\d x_t)^2}{(\int u^p\d x_t)^2}.\label{dh2}
\end{equation}}
\end{lemma}
\begin{proof}
By the definition of $p$-R\'{e}nyi entropy \eqref{I1}, we have
{\scriptsize
\begin{equation*}\begin{aligned}
\frac{\d H_p(u)}{\d t}&=\frac{p}{1-p}\frac{\int u^{p-1}\frac{\partial u}{\partial t}\d x_t}{\int u^p \d x_t}
=\frac{p}{1-p}\frac{\int u^{p-1}\Delta u^p\d x}{\int u^p \d x},\\
\frac{\d^2 H_p(u)}{\d^2t}&=\frac{p}{1-p}\frac{\frac{\partial}{\partial t}(\int u^{p-1}\frac{\partial u}{\partial t}\d x_t)\int u^p\d x_t-\int u^{p-1}\frac{\partial u}{\partial t}\d x_t\frac{\partial}{\partial t}(\int u^{p}\d x_t)}{(\int u^p\d x_t)^2}\\
&\!\!\!\!=\frac{p}{1-p}\frac{\frac{\partial}{\partial t}(\int u^{p-1}\frac{\partial u}{\partial t}\d x_t)\int u^p\d x_t-\int u^{p-1}\frac{\partial u}{\partial t}\d x_t\int pu^{p-1}\frac{\partial u}{\partial t}\d x_t}{(\int u^p\d x_t)^2}\\
&\!\!\!\!=\frac{p}{1-p}\frac{\frac{\partial}{\partial t}(\int u^{p-1}\frac{\partial u}{\partial t}\d x_t)\int u^p\d x_t-p(\int u^{p-1}\frac{\partial u}{\partial t}\d x_t)^2}{(\int u^p\d x_t)^2}.
\end{aligned}\end{equation*}
}
\end{proof}
\begin{lemma}
We have
{\footnotesize
\begin{eqnarray}\label{T1}
\int u^{p-1}\Delta u^p\d x_t=\int\Delta u^{p-1} u^p\d x_t.
\end{eqnarray}}
\end{lemma}
\begin{proof} Integrating by parts\cite{Savare2014}, we have
{\footnotesize
$$\int u^{p-1}\Delta u^p\d x_t=-\int\nabla u^{p-1}\nabla u^p\d x_t=\int\Delta u^{p-1} u^p\d x_t.$$
}
\end{proof}

\begin{lemma}
By Cauchy-Schwarz inequality we have
{\footnotesize
\begin{eqnarray}\label{T2}
(\int \Delta u^{p-1} u^p\d x_t)^2\leq\int u^p\d x_t\int(\Delta u^{p-1})^2u^p\d x_t.
\end{eqnarray}}

\end{lemma}

Then, we obtain
{%\small
\begin{equation}\label{dN2}
\begin{array}{ll}
\frac{\d^2}{\d^2t}N_p(u)&\!\!\!=\frac{\mu}{n}\frac{\d^2H_p(u)}{\d^2t}{\rm e}^{\frac{\mu}{n}H_p(u)}+\left(\frac{\mu}{n}\frac{\d H_p(u)}{\d t}\right)^2{\rm e}^{\frac{\mu}{n}H_p(u)}\\[0.2cm]
&\!\!\!=\frac{\mu}{n}{\rm e}^{\frac{\mu}{n}H_p(u)}I_{2,n},
\end{array}\end{equation}}
where
{\footnotesize $I_{2,n}=\frac{\d^2H_p(u)}{\d^2t}+\dfrac{\mu}{n}\left(\frac{\d H_p(u)}{\d t}\right)^2$}.
So, by \eqref{dh1},\eqref{dh2},we have
{%\tiny
\begin{equation}\label{T33}
\begin{aligned}
I_{2,n}&=\frac{p}{1-p}\frac{\frac{\partial}{\partial t}(\int u^{p-1}\frac{\partial u}{\partial t}\d x_t)\int u^p\d x_t-p(\int u^{p-1}\frac{\partial u}{\partial t}\d x_t)^2}{(\int u^p\d x_t)^2}\\
&+\frac{\mu}{n}(\frac{p}{1-p}\frac{\int u^{p-1}\Delta u^p\d x_t}{\int u^p\d x_t})^2\\
&=\frac{\mu p^2}{n(1-p)^2}\frac{(\int u^{p-1}\Delta u^p\d x_t)^2}{(\int u^p\d x_t)^2}+\frac{p}{1-p}\frac{\frac{\partial}{\partial t}(\int u^{p-1}\frac{\partial u}{\partial t}\d x_t)\int u^p\d x_t}{(\int u^p\d x_t)^2}\\
&-\frac{p^2}{1-p}\frac{(\int u^{p-1}\Delta u^p\d x_t)^2}{(\int u^p\d x_t)^2}\\
&=(\frac{\mu p^2}{n(1-p)^2}-\frac{p^2}{1-p})\frac{(\int u^{p-1}\Delta u^p\d x_t)^2}{(\int u^p\d x_t)^2}\\
&+\frac{p}{1-p}\frac{\frac{\partial}{\partial t}(\int u^{p-1}\frac{\partial u}{\partial t}\d x_t)\int u^p\d x_t}{(\int u^p\d x_t)^2}\\
&\overset{\eqref{T1}}{=}\frac{(\mu-n(1-p))p^2}{n(1-p)^2}\frac{(\int \Delta u^{p-1}u^p\d x_t)^2}{(\int u^p\d x_t)^2}
+\frac{p}{1-p}\frac{\frac{\partial}{\partial t}(\int u^{p-1}\frac{\partial u}{\partial t}\d x_t)\int u^p\d x_t}{(\int u^p\d x_t)^2}\\
&\overset{(i)}{\leq}\frac{(\mu-n(1-p))p^2}{n(1-p)^2}\frac{\int u^p \d x_t\int (\Delta u^{p-1})^2u^p\d x_t}{(\int u^p\d x_t)^2}\\
&+\frac{p}{1-p}\frac{\frac{\partial}{\partial t}(\int u^{p-1}\frac{\partial u}{\partial t}\d x_t)\int u^p\d x_t}{(\int u^p\d x_t)^2}\\
&=\frac{1}{\int u^p\d x_t}\left(\frac{(\mu-n(1-p))p^2}{n(1-p)^2}\int (\Delta u^{p-1})^2u^p\d x_t+\frac{p}{1-p}\frac{\partial}{\partial t}(\int u^{p-1}\frac{\partial u}{\partial t}\d x_t)\right)\\
&=\frac{1}{\int u^p\d x_t}\int\left(\frac{(\mu-n(1-p))p^2}{n(1-p)^2}(\Delta u^{p-1})^2u^p+\frac{p}{1-p}\frac{\partial}{\partial t}(u^{p-1}\frac{\partial u}{\partial t})\right)\d x_t\\
&=\frac{1}{\int u^p\d x_t}\int F_{2,n}\d x_t,
\end{aligned}
\end{equation}
}
where $F_{2,n}=\frac{(\mu-n(1-p))p^2}{n(1-p)^2}(\Delta u^{p-1})^2u^p+\frac{p}{1-p}\frac{\partial}{\partial t}(u^{p-1}\frac{\partial u}{\partial t})$.

{ \textbf{Remark:} In \eqref{T33}, the step $(i)$ is according to \eqref{T2}, and $\frac{(\mu-n(1-p))p^2}{n(1-p)^2}\geq0$ should be satisfied, which is true under condition $p\geq1-\frac{\mu}{n}$. When $\mu:=2+n(p-1)$, $\frac{(\mu-n(1-p))p^2}{n(1-p)^2}\geq0$ yields   $p\ge1-\frac{1}{n}$. Savar\'{e}-Toscani \cite{Savare2014} also used the inequality \eqref{T2}, but ignore the nonnegativity of the coefficient $\frac{(\mu-n(1-p))p^2}{n(1-p)^2}$, thus the parameter's range $p>1-\frac{2}{n}$ in \cite{Savare2014} should be corrected to $p\ge1-\frac{1}{n}$.}

Further, we can get
{%\scriptsize
\begin{equation}\label{T333}
\begin{aligned}
F_{2,n}&=\frac{(\mu-n(1-p))p^2}{n(1-p)^2}u^p\sum\limits_{a=1}^{n}\sum\limits_{b=1}^{n}
[(\frac{\partial^2}{\partial^2x_{a,t}}u^{p-1})(\frac{\partial^2}{\partial^2x_{b,t}}u^{p-1})]\\
&+\frac{p}{1-p}\frac{\partial}{\partial t}[u^{p-1}\sum\limits_{a=1}^{n}(\frac{\partial^2}{\partial^2x_{a,t}}u^{p})]\\
&=\frac{(\mu-n(1-p))p^2}{n(1-p)^2}u^p\sum\limits_{a=1}^{n}\sum\limits_{b=1}^{n}
[(\frac{\partial^2}{\partial^2x_{a,t}}u^{p-1})(\frac{\partial^2}{\partial^2x_{b,t}}u^{p-1})]\\
&+\frac{p}{1-p}\sum\limits_{a=1}^{n}[(p-1)u^{p-2}(\frac{\partial^2}{\partial^2x_{a,t}}u^p)\frac{\partial u}{\partial t}
+pu^{p-1}\frac{\partial^2}{\partial^2x_{a,t}}(u^{p-1}\frac{\partial u}{\partial t})]\\
&=\frac{(\mu-n(1-p))p^2}{n(1-p)^2}u^p\sum\limits_{a=1}^{n}\sum\limits_{b=1}^{n}
[(\frac{\partial^2}{\partial^2x_{a,t}}u^{p-1})(\frac{\partial^2}{\partial^2x_{b,t}}u^{p-1})]\\
&+\frac{p}{1-p}\sum\limits_{a=1}^{n}\left[(p-1)u^{p-2}(\frac{\partial^2}{\partial^2x_{a,t}}u^p)\sum\limits_{b=1}^{n}(\frac{\partial^2}{\partial^2x_{b,t}}u^p)\right.\\
&\left.+pu^{p-1}\frac{\partial^2}{\partial^2x_{a,t}}(u^{p-1}\sum\limits_{b=1}^{n}(\frac{\partial^2}{\partial^2x_{b,t}}u^p))\right]\\
&=\frac{(\mu-n(1-p))p^2}{n(1-p)^2}u^p\sum\limits_{a=1}^{n}\sum\limits_{b=1}^{n}
[(\frac{\partial^2}{\partial^2x_{a,t}}u^{p-1})(\frac{\partial^2}{\partial^2x_{b,t}}u^{p-1})]\\
&+\frac{p}{1-p}\sum\limits_{a=1}^{n}\sum\limits_{b=1}^{n}\left[(p-1)u^{p-2}(\frac{\partial^2}{\partial^2x_{a,t}}u^p)(\frac{\partial^2}{\partial^2x_{b,t}}u^p)\right.\\
&\left.+pu^{p-1}\frac{\partial^2}{\partial^2x_{a,t}}\left(u^{p-1}(\frac{\partial^2}{\partial^2x_{b,t}}u^p)\right)\right]\\
&=\sum\limits_{a=1}^{n}\sum\limits_{b=1}^{n}\mathcal{T}_{a,b},
\end{aligned}\end{equation}}
where
{%\scriptsize
\begin{equation}\label{T334}
\begin{aligned}
\mathcal{T}_{a,b}&=\frac{(\mu-n(1-p))p^2}{n(1-p)^2}u^p
[(\frac{\partial^2}{\partial^2x_{a,t}}u^{p-1})(\frac{\partial^2}{\partial^2x_{b,t}}u^{p-1})]\\
&+\frac{p}{1-p}\left[(p-1)u^{p-2}(\frac{\partial^2}{\partial^2x_{a,t}}u^p)(\frac{\partial^2}{\partial^2x_{b,t}}u^p)\right.\\
&\left.+pu^{p-1}\frac{\partial^2}{\partial^2x_{a,t}}\left(u^{p-1}(\frac{\partial^2}{\partial^2x_{b,t}}u^p)\right)\right].
\end{aligned}\end{equation}}
For convenience, introduce the notation
$u_{i,j}:=\frac{\partial^{i+j}u}{\partial^ix_{a,t}\partial^j x_{b,t}}$.
%{%\scriptsize
% \begin{equation}\begin{array}{ll}
%u_{4,0}:=\frac{\partial^4u}{\partial^4x_{a,t}},u_{3,1}:=\frac{\partial^4u}{\partial^3x_{a,t}\partial x_{b,t}},
% u_{2,2}:=\frac{\partial^4u}{\partial^2x_{a,t}\partial^2 x_{b,t}},\\
% u_{1,3}:=\frac{\partial^4u}{\partial x_{a,t}\partial^3 x_{b,t}},u_{0,4}:=\frac{\partial^4u}{\partial^4 x_{b,t}},
% u_{3,0}:=\frac{\partial^3u}{\partial^3x_{a,t}},\\
% u_{2,1}:=\frac{\partial^3u}{\partial^2x_{a,t}\partial x_{b,t}},
% u_{1,2}:=\frac{\partial^3u}{\partial x_{a,t}\partial^2 x_{b,t}},u_{0,3}:=\frac{\partial^3u}{\partial x_{b,t}^3},\\
% u_{2,0}:=\frac{\partial^2u}{\partial^2x_{a,t}},
% u_{1,1}:=\frac{\partial^2u}{\partial x_{a,t}\partial x_{b,t}},
% u_{0,2}:=\frac{\partial^2u}{\partial^2 x_{b,t}},\\
% u_{1,0}:=\frac{\partial u}{\partial x_{a,t}},u_{0,1}:=\frac{\partial u}{\partial x_{b,t}}.
% \end{array}\end{equation}}
By simple computation, we have {\footnotesize $\mathcal{T}_{a,b}=-\frac{u^{3p-6}p^2}{(p-1)n}T_{a,b}$}, and
{
\begin{equation}
\label{eq-tab}
\begin{aligned}
T_{a,b}&=4np^4u_{0, 1}^2u_{1, 0}^2+2np^3uu_{0, 1}^2u_{2, 0}
+8np^3uu_{0, 1}u_{1, 0}u_{1, 1}\\&
+4np^3uu_{0, 2}u_{1, 0}^2-15np^3u_{0, 1}^2u_{1, 0}^2
-\mu p^3u_{0, 1}^2u_{1, 0}^2+2np^2u^2u_{0, 1}u_{2, 1}\\&
+2np^2u^2u_{0, 2}u_{2, 0}
+4np^2u^2u_{1, 0}u_{1, 2}+2np^2u^2u_{1, 1}^2-3np^2uu_{0, 1}^2u_{2, 0}\\&
-20np^2uu_{0, 1}u_{1, 0}u_{1, 1}-8np^2uu_{0, 2}u_{1, 0}^2-\mu p^2uu_{0, 1}^2u_{2, 0}
-\mu p^2uu_{0, 2}u_{1, 0}^2\\&
+16np^2u_{0, 1}^2u_{1, 0}^2+5\mu p^2u_{0, 1}^2u_{1, 0}^2
+npu^3u_{2, 2}-2npu^2u_{0, 1}u_{2, 1}\\&
-npu^2u_{0, 2}u_{2, 0}-4npu^2u_{1, 0}u_{1, 2}
-2npu^2u_{1, 1}^2-\mu pu^2u_{0, 2}u_{2, 0}\\&
-npuu_{0, 1}^2u_{2, 0}
+12npuu_{0, 1}u_{1, 0}u_{1, 1}+2npuu_{0, 2}u_{1, 0}^2+3\mu puu_{0, 1}^2u_{2, 0}\\&
+3\mu puu_{0, 2}u_{1, 0}^2-npu_{0, 1}^2u_{1, 0}^2-8\mu pu_{0, 1}^2u_{1, 0}^2-nu^2u_{0, 2}u_{2, 0}\\&
+\mu u^2u_{0, 2}u_{2, 0}+2nuu_{0, 1}^2u_{2, 0}+2nuu_{0, 2}u_{1, 0}^2-2\mu uu_{0, 1}^2u_{2, 0}\\&
-2\mu uu_{0, 2}u_{1, 0}^2-4nu_{0, 1}^2u_{1, 0}^2+4\mu u_{0, 1}^2u_{1, 0}^2,
\end{aligned}\end{equation}}
which is a 4th-order differential form.

From  \eqref{dN2}, \eqref{T33}, \eqref{T333}, and \eqref{T334}, we have
{%\scriptsize
\begin{equation}\label{E2n}
\frac{\d^2}{\d^2t}N_p(u)\le-\frac{p^2\mu}{n^2}e^{\frac{\mu}{n}H_p(u)}\frac{1}{\int u^p\d x_t}\int u^{3p-6}
E_{2,n}\d x_t
\end{equation}}
where {$E_{2,n}=\sum\limits_{a=1}^{n}\sum\limits_{b=1}^{n}\frac{T_{a,b}}{p-1}$}
and $T_{a,b}$ is defined in \eqref{eq-tab}.
Then the problem {$\frac{\d^2}{\d^2t}N_p(u)\leq0$} can be transformed to {  $\int u^{3p-6}E_{2,n}\d x_t\geq0$}.
Thus, Lemma \ref{lm-pr1} is proved.

\section{A generalized version of CREP}
\label{sec-repi}
In this section, we prove a generalized CREP using the procedure given in section \ref{sec-2}.
\begin{theorem}
\label{th-main1}
Let $u(x_t)$ be a probability density in $\R^n$ solving \eqref{HeatEqu2} and satisfying \eqref{lem4}. Then we give a formula $\Phi(n,p,\mu)$ such that the $p$-th R\'{e}nyi entropy power defined in \eqref{I2} satisfies
\begin{eqnarray}
\label{eq-soi1}
\frac{\d^2}{\d^2t}N_p(x_t)\leq0,
\end{eqnarray}
under the condition $\Phi(n,p,\mu)$, that is $N_p(x_t)$ is concave under $\Phi(n,p,\mu)$.
\end{theorem}

The proof of the above theorem consists three steps which will be given in the following subsections.

\subsection{Reduce to a finite problem}

We first give an inequality  constraint.
Denote that $|\nabla^2f|^2=\sum_{i,j}(\frac{\partial^2f}{\partial x_{i}\partial x_{j}})^2$.
Then, based on the trace inequality $|\nabla^2f|^2\geq\frac{1}{n}(\Delta f)^2$, we give a nonnegative constraint:
{%\footnotesize
\begin{equation}
\label{eq-nc1}
\begin{aligned}
I_1=\frac{u^p}{u^{3p-6}}\left[|\nabla^2u^{p-1}|^2-\frac{1}{n}(\Delta u^{p-1})^2\right]
=\sum\limits_{a=1}^{n}\sum\limits_{b=1}^{n}I_{1,a,b}\ge0
\end{aligned}
\end{equation}}
where {  $I_{1,a,b}=u^{6-2p}\left[(\frac{\partial^2u^{p-1}}{\partial x_{a,t}\partial x_{b,t}})^2-\frac{1}{n}\frac{\partial^2u^{p-1}}{\partial^2x_{a,t}}\frac{\partial^2u^{p-1}}{\partial^2x_{b,t}}\right]$}.

From \eqref{E2n} and \eqref{eq-nc1}, in order for \eqref{eq-soi1} to be true,
it suffices to solve
\\[0.1cm]
\noindent{\bf Problem I}.
Find a formula $\Phi(n,p,\mu)$  such that
{%\footnotesize
$$
\begin{array}{l}
{E}_{2,n}\ge\widetilde{E}_{2,n}=E_{2,n}+c_1I_1=\sum\limits_{a=1}^{n}\sum\limits_{b=1}^{n}(\frac{1}{p-1}T_{a,b} + c_1I_{1,a,b}) \ge0
\end{array}
$$}
under the conditions $c_1\le0$, $p\ge1-\frac{\mu}{n}$, {$R_{i,a,b}=0,i=1,\ldots,28$} given in \eqref{eq-2cons}.

%Motivated by symmetric functions, we have
Since
{\small $\sum\limits_{a=1}^{n}\sum\limits_{b=1}^{n}T_{a,b}=\sum\limits_{a=1}^{n}\sum\limits_{b=1}^{n}T_{b,a}$} and $I_{1,a,b}=R^{(I)}_{1,b,a}$, we have
{%footnotesize
\begin{equation}
\label{eq-Lab}
\begin{aligned}
\widetilde{E}_{2,n}&=\frac{1}{2}\sum\limits_{a=1}^{n}\sum\limits_{b=1}^{n}[\frac{1}{p-1}
  (T_{a,b}+T_{b,a})+c_1(I_{1,a,b}+I_{1,b,a})]=\frac{1}{2}\sum\limits_{a=1}^{n}\sum\limits_{b=1}^{n}L_{a,b},
\end{aligned}\end{equation}}
where {\small$L_{a,b}=\frac{1}{p-1}(T_{a,b}+T_{b,a})+c_1(I_{1,a,b}+I_{1,b,a})$}.

\subsection{Simplify the problem with the constraints}

From \eqref{eq-Lab}, to solve {\bf Problem I}, it suffices to solve
\\[0.1cm]
\noindent{\bf Problem II}.
Find a formula $\Phi(n,p,\mu)$ such that $L_{a,b}\ge0$ under the  conditions $c_1\le0$, $p\ge1-\frac{\mu}{n}$, and $R_{i,a,b}=0,i=1,\ldots,{28}$.

In this section, we simplify $L_{a,b}$ in {\bf Problem II} with the constraints.
Note that the subscripts $a$ and $b$ are fixed and will be treated as symbols.
%and we can solve {\bf Problem II} with  Procedure \ref{proc-H}
%with $L_{a,b}$ and $R_{i,a,b},i=1,\ldots,{28}$ as inputs.

Our goal is to reduce $L_{a,b}$ into a quadratic form in certain new variables.
The new variables are all the monomials in $\R[{\mathcal V}_{a,b}]$
with degree 2 and total order 2, where ${\mathcal V}_{a,b}$ is defined in \eqref{eq-v2}.
{%\scriptsize
$$
\begin{aligned}
&
m_1=uu_{0,2},\
 m_2=uu_{1,1},\
m_3=uu_{2,0}\\
&
m_4=u_{0,1}^2,\
m_5=u_{1,0}u_{0,1},\
m_6=u_{1,0}^2.
%\mathfrak{s}=
%&
%m_1=u\frac{\partial^2u}{\partial^2x_{b,t}},\
% m_2=u\frac{\partial^2u}{\partial x_{a,t}\partial x_{b,t}},\
%m_3=u\frac{\partial^2u}{\partial^2 x_{a,t}}\\
%&
%m_4=\left(\frac{\partial u}{x_{b,t}}\right)^2,\
%m_5=\frac{\partial u}{\partial x_{a,t}}\frac{\partial u}{x_{b,t}},\
%m_6=\left(\frac{\partial u}{x_{a,t}}\right)^2.
\end{aligned}
$$}

We will simplify the constraints in \eqref{eq-2cons} as follows.
A quadratic monomial in $m_i$ is called a {\em quadratic monomial}.
Write monomials in $\C_{2,n}=\{R_i,i=1,\ldots,N_1\}$ as quadratic monomials if possible.
Doing Gaussian elimination to $\C_{2,n}$ by treating the monomials
as variables and according to a monomial order such that a quadratic monomial is less than a non-quadratic monomial, we obtain
$$\widetilde{\C}_{2,n}={ \C}_{2,n,1}\cup { \C}_{2,n,2},$$
where ${ {\C}}_{2,n,1}$ is the set of quadratic forms in $m_i$,
${{\C}}_{2,n,2}$ is the set of non-quadratic forms, and $\span_\R(\C_{2,n})=\span_\R(\widetilde{\C}_{2,n})$.
We  obtain ${\C}_{2,n,1}=\{\widehat{R}_i,i=1,\ldots,9\}$
and ${\C}_{2,n,2}=\{\widetilde{R}_i,i=1,\ldots,13\}$, where

{%\scriptsize
$$\begin{array}{ll}
\widehat{R}_1=2m_{1}m_{5}+\frac{2(3p-5)}{3}m_{4}m_{5},
\widehat{R}_2=m_{2}m_{6}+\frac{3p-5}{3}m_{5}m_{6},\\
\widehat{R}_3=-6m_{3}m_{5}+2(5-3p)m_{5}m_{6},
\widehat{R}_4=(3p-5)m_{4}^2+3m_{1}m_{4},\\
\widehat{R}_5=(3p-5)m_{6}^2+3m_{3}m_{6},
\widehat{R}_6=(3p-5)m_{4}m_{5}+3m_{2}m_{4},\\
\widehat{R}_7=(3p-5)m_{5}^2+2m_{2}m_{5}+m_{3}m_{4}, \\
\widehat{R}_8=m_{1}m_{3}-m_{2}^2+\frac{9p-12}{2}m_{3}m_{4}+\frac{9p^2-27p+20}{2}m_{5}^2, \\
\widehat{R}_9=m_{1}m_{6}-m_{3}m_{4}.\\
%\end{array}$$}
%
%$$\begin{array}{ll}
\widetilde{R}_{1}=u^3u_{0,4}+(3-3p)m_{1}^2+(9p^3-36p^2+47p-20)m_{4}^2,\\
\widetilde{R}_{2}=u^3u_{1,3}+(3-3p)m_{1}m_{2}+(9p^3-36p^2+47p-20)m_{4}m_{5},\\
\widetilde{R}_{3}=u^3u_{3,1}+(3-3p)m_{2}m_{3}+(-9p^2+21p-12)m_{3}m_{5},\\
\widetilde{R}_{4}=u^3u_{4,0}+(3-3p)m_{3}^2+(9p^3-36p^2+47p-20)m_{6}^2,\\
\widetilde{R}_{5}=u^2u_{0,1}u_{0,3}+m_{1}^2+\frac{-9p^2+27p-20}{3}m_{4}^2,\\
\widetilde{R}_{6}=u^2u_{0,1}u_{1,2}+m_{1}m_{2}+\frac{-9p^2+27p-20}{3}m_{4}m_{5},\\
\widetilde{R}_{7}=u^2u_{0,1}u_{3,0}+m_{2}m_{3}+\frac{-9p^2+27p-20}{3}m_{5}m_{6},\\
\widetilde{R}_{8}=u^2u_{1,0}u_{0,3}+m_{1}m_{2}+\frac{-9p^2+27p-20}{3}m_{4}m_{5},\\
\widetilde{R}_{9}=u^2u_{1,0}u_{2,1}+m_{2}m_{3}+\frac{-9p^2+27p-20}{3}m_{5}m_{6},\\
\widetilde{R}_{10}=u^2u_{1,0}u_{3,0}+m_{3}^2+\frac{-9p^2+27p-20}{3}m_{6}^2,\\
\widetilde{R}_{11}=u^3u_{2,2}+(3-3p)m_{2}^2+\frac{9p^2-21p+12}{2}m_{3}m_{4}+\frac{27p^3-108p^2+141p-60}{2}m_{5}^2,\\
\widetilde{R}_{12}=u^2u_{0,1}u_{2,1}+m_{2}^2+\frac{4-3p}{2}m_{3}m_{4}+\frac{-9p^2+27p-20}{2}m_{5}^2, \\
\widetilde{R}_{13}=u^2u_{1,0}u_{1,2}+m_{2}^2+\frac{4-3p}{2}m_{3}m_{4}+\frac{-9p^2+27p-20}{2}m_{5}^2.
\end{array}$$

%Step {\bf S3}. There exists one intrinsic constraint: $\widehat{R}_{10}=m_{4}m_{6}-m_{5}^2$ and $N_{3}=10$.
%We do not need Step {\bf S4}, since there exist no log-concave constraints.

We now simplify $L_{a,b}$ using $\C_{2,n,1}$ and $\C_{2,n,2}$.
Eliminating the non-quadratic monomials in $L_{a,b}$ using $\C_{2,n,2}$, and doing further reduction by $\C_{2,n,1}$, we have
{\scriptsize
\begin{equation}
\label{2.12a}
\begin{aligned}
&\widehat{L}_{a,b}=L_{a,b}-2(p^3c_{1}+4np^2-4p^2c_{1}-6np+5pc_{1}-2c_{1})\widehat{R}_{7}\\
&-\frac{2}{n}(2n^2p-p^2c_{1}+n^2-n\mu +2pc_{1}-c_{1})\widehat{R}_{8}\\
&-\frac{1}{n}(6n^2p^2-2p^3c_{1}-5n^2p-2np\mu +8p^2c_{1}-4n^2+4n\mu -10pc_{1}+4c_{1})\widehat{R}_{9}\\
&-\frac{2np}{p-1}\widetilde{R}_{11}
-6np\widetilde{R}_{12}
-6np\widetilde{R}_{13}\\
&=(2np+2n-2\mu )m_{2}^2+(5np-5np^2+5p\mu +4n-4\mu )m_{3}m_{4}\\
&+(18np^2-7np^3+7p^2\mu -3np-19p\mu -12n+12\mu )m_{5}^2\\
&+\frac{c_{1}}{n}[{(2n-2)(p^2-2p+1)}m_{2}^2+(4n-2np+5p-4)(p^2-2p+1)m_{3}m_{4}\\
&+(14np-4np^2+7p^2-12n-19p+12)(p^2-2p+1)m_{5}^2]
\end{aligned}\end{equation}}

In order for $\widehat{L}_{a,b}\ge0$ to be true, we need to eliminate
the monomial $m_{3}m_{4}$ from $\widehat{L}_{a,b}$, which can be done with $\widehat{R}_{7}$ as follows.
\begin{equation}
\label{eq-or1}
\widehat{L}_{a,b}+p_{7}\widehat{R}_{7}=A_1m_2^2+A_2m_2m_5+A_3m_5^2
\end{equation}
where
{\begin{equation*}
\begin{aligned}
&p_{7}=(2np^3c_{1}+5n^2p^2-8np^2c_{1}-5p^3c_{1}-5n^2p-5np\mu +10npc_{1}\\
&\quad +14p^2c_{1}-4n^2+4n\mu -4nc_{1}-13pc_{1}+4c_{1})/n.\\
&A_1=-2c_{1}p^2/n+4c_{1}p/n+2np+2c_{1}+2c_{1}p^2-4c_{1}p-2c_{1}/n-2\mu +2n,\\
&A_2=4c_{1}p^3-16c_{1}p^2-10c_{1}p^3/n-10p\mu +20c_{1}p+28c_{1}p^2/n-26c_{1}p/n\\
&\quad+10np^2-10np+8\mu -8c_{1}+8c_{1}/n-8n,\\
&A_3=-8\mu +8n+26c_{1}p^2-24c_{1}p-8c_{1}/n-12c_{1}p^3+2c_{1}p^4+8c_{1}-52c_{1}p^2/n\\
&\quad+34c_{1}p/n+34c_{1}p^3/n-8c_{1}p^4/n+18p\mu +8np^3-22np^2-8p^2\mu +10np.
\end{aligned}
\end{equation*}}

\subsection{Compute $\Phi(n,p,\mu)$ }
From \eqref{eq-or1}, in order to solve \textbf{Problem II}, it suffices to solve

\noindent\textbf{Problem III}:
Find a formula $\Phi(n,p,\mu)$  such that
\begin{equation}
\label{eq-A1}
\Phi(n,p,\mu) = \exists c_1 \forall m_1,m_2,m_3 (c_1\le0 \wedge p\ge1-\frac{\mu}{n} \wedge A_1m_2^2+A_2m_2m_5+A_3m_5^2\ge0).
\end{equation}
%under the condition $c_1\le0$ and $p\ge1-\frac{\mu}{n}$.

In principle, \textbf{Problem III} can be solved with the quantifier elimination~\cite{QE}.
In this paper, the problem is special and an explicit proof is given.

By the knowledge of linear algebra, we have $A_1m_2^2+A_2m_2m_5+A_3m_5^2\ge0$ is equivalent to
$\Delta_{1}=A_1=\frac{2}{n}s_1\ge0$, $\Delta_{2}=A_3=\frac{2}{n}s_{2}\ge0$,
$\Delta_{3}=A_1A_3-\frac{1}{4}A_2^2=\frac{p}{n^2}s_3\ge0$,
where
{%\scriptsize
%\color{red}
\begin{equation*}\begin{aligned}
&s_{1}=(p-1)^2(n-1)c_{1}+n^2(p+1)-n\mu ,\\
&s_{2}=(p-1)^2(n(p-2)^2-4p^2+9p-4)c_{1}\\
&+n^2(4p^3-11p^2+5p+4)-(4p^2-9p+4)n\mu ,\\
%&s_{3}=-(p-1)^4(4n+9p-4)c_{1}^2+2n(p-1)^2(2n^2p+9np^2-2n^2-13np\\
%&-2n\mu -9p\mu +4\mu )c_{1}-n^2(9np^2-13np-9p\mu -4n+4\mu )(np-n-\mu ).\\
&s_{3}=(4-9p)n^2(\mu -\mu_3)(\mu -\mu_4),
%-(p-1)^4(4n+9p-4)c_{1}^2+2n(p-1)^2(2n^2p+9np^2-2n^2-13np\\
%&-2n\mu -9p\mu +4\mu )c_{1}-n^2(9np^2-13np-9p\mu -4n+4\mu )(np-n-\mu ).
\end{aligned}\end{equation*}}
where $\mu_3$ and $\mu_4$ are defined in \eqref{eq-mu1}
and $p\neq \frac{4}{9}$ is assumed in $\mu_4$.
We thus proved
\begin{lemma} We can eliminate $m_1,m_2,m_3$ in \eqref{eq-A1}:
\begin{equation}\label{set}
\Phi(n,p,\mu)=\exists c_1(s_1\ge0 \wedge  s_2\ge0\wedge  s_3\ge0\wedge c_1\le0\wedge  p\geq1-\mu/n).
\end{equation}
\end{lemma}
We will give an explicit formula for $\Phi$ in \eqref{set}.
First, introduce the following parameters.

{%\scriptsize
\footnotesize
\begin{equation}
\label{eq-phi0}
\begin{array}{lll}
n_1 = \frac{9-\sqrt{17}}{8},
&n_2=\frac{9+\sqrt{17}}{8},
&n_3=(\sqrt{17}+1)/2,\\
\theta_1=-\frac{2n}{(p-1)^2},
&\theta_2=\frac{2n^2p(9p-13)}{(p-1)^2(4n+9p-4)},
&\theta_3=(\sqrt{17}-9)n,\\
\theta_4=\frac{n^2p-n^2-n\mu}{(p-1)^2},
&\theta_5=\frac{n^2(9p^2-13p-4)-n(9p-4)\mu}{(p-1)^2(4n+9p-4)},
&\theta_6=\frac{-4n(\sqrt{17}-1)\mu+8n^2}{\sqrt{17}+1},\\
\theta_7 = n(1-p),
&\theta_8 =5n/9,
&\theta_9=-\frac{162n}{25},\\
\theta_{10}=\frac{64n}{\sqrt{17}-9},
&\theta_{11}=\frac{8(9\sqrt{17}+23)n^2}{-32n-49+9\sqrt{17}},
&\theta_{12}=-\frac{16(11\sqrt{17}+47)n\mu+152n^2}{26\sqrt{17}n+73\sqrt{17}+118n+305},\\
\theta_{13}=-\frac{4n(\mu\sqrt{17}+2n+\mu)}{\sqrt{17}-1},
&\theta_{14}=-\frac{8n(22\mu\sqrt{17}-19n-94\mu)}{26\sqrt{17}n+73\sqrt{17}-118n-305},
&\theta_{15}=-\frac{9}{5}n^2-\frac{81}{25}\mu n,\\
\theta_{16} = (\sqrt{17}-1)n/8,
& &\\
 \phi_1\triangleq p\ge1-\frac{1}{n},
& \phi_2 \triangleq \mu =2+n(p-1),
&\phi_3 \triangleq p>1-\frac{1}{n}.
\\
\end{array}
\end{equation}
}

\begin{table}[ht]
\centering
\scalebox{0.8}{
\begin{tabular}{|c|c|c|c|}               \hline
$p>n_2$
& $\theta_4> \theta_1\wedge \theta_5\le0$
& *
& $ \phi_1 \wedge \phi_2$
\\ \hline
$p=n_2$
& $\theta_6> \theta_3 \wedge \theta_{12}\le0$
& *
& $ \phi_1 \wedge \phi_2$
\\ \hline
$\frac{13}{9}< p<n_2$
& $\theta_4> \theta_1\wedge \theta_5\le0$
& *
& $ \phi_1 \wedge \phi_2$
\\ \hline
$n_1<p\leq\frac{13}{9}$, $p\neq1$
& $\phi_3 \wedge \theta_4> \theta_1\wedge \theta_5\le \theta_2$
& $\phi_3 \wedge \theta_7 \le \mu \wedge \theta_4 > \theta_2$
& $ \phi_1 \wedge \phi_2$
\\ \hline
$p=n_1$
& $n< n_3 \wedge \theta_{13}> \theta_{10} \wedge \theta_{14}\le \theta_{11}$
& $n< n_3 \wedge \theta_{13}> \theta_{11} \wedge \mu\ge \theta_{16}$
&$ \phi_1 \wedge \phi_2$
\\ \hline
$\frac{4}{9}<p<n_1$
& $\phi_3 \wedge \theta_4> \theta_1\wedge \theta_5\le \theta_2$
& $\phi_3 \wedge \theta_7 \le \mu \wedge \theta_4 > \theta_2$
& $\phi_1 \wedge \phi_2$
%&$p\neq1 \wedge \phi_1 \wedge \phi_2$
\\ \hline
$p=\frac{4}{9}$
& $n=1 \wedge \theta_{15}> \theta_9 \wedge \mu\ge \theta_8$
& *
& $n=1\wedge -p\le\mu-1\le p$
\\ \hline
$0<p<\frac{4}{9}$
& $n=1 \wedge \theta_7 \le \mu \wedge \theta_4 > \theta_1$
& $n=1 \wedge \theta_5< \theta_1 \wedge \mu\ge \theta_7$
& $n=1 \wedge -p\le\mu-1\le p$
\\ \hline
 \end{tabular}}
\caption{The description for $\Phi$ in  \eqref{set}}\label{tab11}
\end{table}
With these notations, we introduce the conditions for 
defining $\Phi$ in Table  \ref{tab11},
where $*$ means $\varnothing$.
Define $T(i,j)$ to be the formula in the $i$-th row and the $j$-th column in Table \ref{tab11}.
Then we denote
\begin{equation}
\label{eq-tij}
\mathbb{T}(i,j) \triangleq   T(i,1)   \wedge T(i,j)
\hbox{ for } i=1,\ldots,8, j=2,\ldots, 4.
\end{equation}
For examples, $\mathbb{T}(1,2)=(p>n_2\wedge \theta_4> \theta_1\wedge \theta_5\le0)$,
which means that if $p,n,\mu$ satisfy  $\mathbb{T}(1,2)$ then
there exists a $c_1\le0$ such that \eqref{eq-A1} is true and the CREP is valid.
$\mathbb{T}(1,3)= \varnothing$, which means that there exist no values
for $p,n,\mu$ such that \eqref{eq-A1} and the CREP are true.
%Also note that, the value ranges of $c_1$ in the first column of the table are not disjoint.

We now give the main result of the paper, which implies Theorem \ref{th-main1}.
The proof for the theorem can be found in section \ref{sec-proof2}.
Also, from the proof we can see that for the
cases $\mathbb{T}(i,4),i=1,\ldots,8$, the corresponding value for $c_1$ 
is $c_1=\theta_1$. 
\begin{theorem}
\label{nsconditions}
The sufficient and necessary condition for \textbf{Problem III},
 that is, \eqref{eq-A1} to be true,
 %under the condition  $c_1\le0$ and $p\ge1-\frac{\mu}{n}$,
is
$$\Phi(n,p,\mu)=\vee_{i=1}^{8}\vee_{j=2}^{4} \mathbb{T}(i,j),$$
where $\mathbb{T}(i,j)$ is defined in \eqref{eq-tij}
and  $\vee$ means  disjunction.
\end{theorem}

\subsection{Compare with existing results}

We will show that  our result includes the result proved in \cite{Savare2014}, and essentially more results.

In \cite{Savare2014},  CREP was proved under the conditions $\mu=2+n(p-1)$ and
$p\geq1-\frac{1}{n}$.
Obviously, the result proved in \cite{Savare2014} corresponds to $\mathbb{T}(i,4),i=1,\ldots,8$ in Table \ref{tab11}.

We can also prove the result   in \cite{Savare2014} directly as follows.
Set $\mu=2+n(p-1)$ and $c_{1}=\theta_1\le0$ in \eqref{2.12a},
we obtain $\widehat{L}_{a,b}=0$.
Also,  the condition $p\geq1-\frac{\mu}{n}$ implies $p\geq1-\frac{1}{n}$.
So, when $\mu=2+n(p-1)$ and $p\geq1-\frac{1}{n}$, the CREP is proved based on our proof procedure.

We can use the SDP code in \cite{GYG2020}[APPENDIX B] to verify the result in Table \ref{tab11}.
For instance,  for $\mu=2, p=\frac{11}{5}, n=2$, the condition $p\geq1-\frac{\mu}{n}$ is satisfied naturally. With the SDP code in \cite{GYG2020}, we obtain
$\widehat{L}_{a,b}+\frac{172}{25}\widehat{R}_{7}=(2\sqrt{2}m_{2}+\frac{344}{100\sqrt{2}}m_{5})^2+\frac{22}{625}m_{5}^2\geq0$ with $c_{1}=-\frac{5}{9}$.
Thus, the CREP is proved when $\mu=2, p=\frac{11}{5}, n=2$.
This case $[ \mu=2,p=\frac{11}{5},n=2, c_{1}=-\frac{5}{9}]$ is included in $\mathbb{T}(1,1)$ in Table \ref{tab11}.
Note that $\mu=2+n(p-1)$ is not satisfied for these parameters.

\subsection{Proof of Theorem \ref{nsconditions}}
\label{sec-proof2}

Introduce more parameters.
{%\tiny
\scriptsize
\begin{equation}
\label{eq-mu1}
\begin{aligned}
 \mu_{1} &= ((p-1)^2(n-1)c_1 + n^2(p+1))/n,\\
 \mu_{2} &=((p-1)^2(n(p-2)^2-4p^2+9p-4)c_{1}+n^2(4p^3-11p^2+5p+4))/(n(4p^2-9p+4))\\
 \mu_{3} &= (n^2p-p^2c_{1}-n^2+2p c_{1}-c_{1})/n,\\
 \mu_{4} &= (9n^2p^2-4np^2c_{1}-9p^3c_{1}-13n^2p+8npc_{1}+22p^2c_{1}-4n^2-4nc_{1}-17p c_{1}+4c_{1})/(n(9p-4)),\\
 \mu_{5} &= -(nc_{1}\sqrt{17}-c_{1}\sqrt{17}+136n^2+17nc_{1}-17c_{1})/(4n(\sqrt{17}-17)),\\
 \mu_{6} &= -(c_{1}\sqrt{17}-8n^2+c_{1})/(4n(\sqrt{17}-1)),\\
 \mu_{7} &= -(26nc_{1}\sqrt{17}+73c_{1}\sqrt{17}+152n^2+118nc_{1}+305c_{1})/(16n(11\sqrt{17}+47)),\\
 \mu_{8} &= -(nc_{1}\sqrt{17}-c_{1}\sqrt{17}-136n^2-17nc_{1}+17c_{1})/(4n(\sqrt{17}+17)),\\
 \mu_{9} &= -(c_{1}\sqrt{17}+8n^2-c_{1})/(4n(\sqrt{17}+1)),\\
 \mu_{10} &= -(26nc_{1}\sqrt{17}+73c_{1}\sqrt{17}-152n^2-118nc_{1}-305c_{1})/(16n(11\sqrt{17}-47)),\\
 \mu_{11} &= (117n^2+25nc_{1}-25c_{1})/(81n),\\
 \mu_{12} &= (7218n^2+1225nc_{1}-400c_{1})/(1296n),\\
 \mu_{13} &= -(5(9n^2+5c_{1}))/(81n),\\
% \mu_{14} =5n/9.\\
%c_{13}&=-\frac{2n^2p(2p-3)^2}{(p-1)^2(np^2-4np-4p^2+4n+9p-4)},\ \
\eta_1&=\frac{2n^2}{p-1},
\eta_2=-\frac{16n^2}{\sqrt{17}-1},
\eta_3=-\frac{18}{5}n^2.
\\
%%
%\phi_1&\triangleq p\ge1-\frac{1}{n},\ \phi_2 \triangleq \mu =2+n(p-1).
\end{aligned}\end{equation}
}

%\begin{proof}
We first treat the three inequalities $s_1\ge0$, $s_2\ge0$, $s_3\ge0$.
Firstly,  $s_1\ge0$ is equivalent to $ \mu\le \mu_{1}$.
Secondly, since the roots of $4p^2-9p+4=0$ are $n_1$ and $n_2$,
we have $s_2\ge0 \Leftrightarrow \mu\le \mu_{2}$ if $p<n_1$ or $p>n_2$;
and
$s_2\ge0 \Leftrightarrow\mu\ge \mu_{2}$ if $n_1<p<n_2$.
In order to analyse $s_3\ge0$, we first compute
\begin{equation}
\label{eq-m34}
\mu_{3}-\mu_{4} = \frac{4((p-1)^2c_{1}+2n)}{9p-4}.
\end{equation}
Therefore, $s_3\ge0$ can be divided into four cases:
$s_3\ge0 \Leftrightarrow \mu_{4}\le\mu\le \mu_{3}$ if $p>\frac{4}{9}$ and $\theta_1<c_1$;
$s_3\ge0 \Leftrightarrow \mu_{3}\le\mu\le \mu_{4}$ if $p>\frac{4}{9}$ and $c_1< \theta_1$;
$s_3\ge0 \Leftrightarrow \mu\ge\mu_{3}\ {\rm or}\ \mu\le\mu_{4}$ if $p<\frac{4}{9}$ and $c_1<\theta_1$;
$s_3\ge0 \Leftrightarrow \mu\ge\mu_{4}\ {\rm or}\ \mu\le\mu_{3}$ if $p<\frac{4}{9}$ and $\theta_1<c_1$.
Finally,  $p\ge1-\frac{ \mu}{n}$ is equivalent to $ \mu\ge \theta_7$.

Based on the above analysis and \eqref{set}, $\Phi(n,p,\mu)$
can be divided into six cases:
{%\scriptsize
\footnotesize
\begin{equation}
\label{eq-phi1}
\begin{array}{l}
%\begin{aligned}
 \Phi(n,p,\mu) \Leftrightarrow \max(\mu_{4},\theta_7)\le\mu\le \min(\mu_{1},\mu_{2},\mu_{3}), \hbox{ if } ( p\in(\frac{4}{9},n_1)  \hbox{ or }  p> n_2 ) \hbox{ and } \theta_1<c_1\le0 ; \\
  \Phi(n,p,\mu) \Leftrightarrow\max(\mu_{2},\mu_{4},\theta_7)\le\mu\le \min(\mu_{1},\mu_{3}) , \hbox{ if }  p\in(n_1,n_2)   \hbox{ or }  \theta_1<c_1\le0 ; \\

  \Phi(n,p,\mu) \Leftrightarrow \max(\mu_{3},\theta_7)\le\mu\le \min(\mu_{1},\mu_{2},\mu_{4}) , \hbox{ if } ( p\in(\frac{4}{9},n_1)  \hbox{ or }  p> n_2 ) \hbox{ or } c_1<\theta_1 ; \\
  \Phi(n,p,\mu) \Leftrightarrow \max(\mu_{2},\mu_{3},\theta_7)\le\mu\le \min(\mu_{1},\mu_{4}) , \hbox{ if }  p\in(n_1,n_2)  \hbox{ or }  c_1<\theta_1 ; \\
  \Phi(n,p,\mu) \Leftrightarrow
 \theta_7\le\mu\le\min(\mu_{1},\mu_{2},\mu_{4})  \hbox{ or }
 \max(\mu_{3},\theta_7)\le \mu \le\min(\mu_{1},\mu_{2}) , \hbox{ if }  p<\frac{4}{9}  \hbox{ or }  c_1< \theta_1 ; \\
  \Phi(n,p,\mu) \Leftrightarrow
 \theta_7\le\mu\le\min(\mu_{1},\mu_{2},\mu_{3})  \hbox{ or }
 \max(\mu_{4},\theta_7)\le \mu\le\min(\mu_{1},\mu_{2}) , \hbox{ if }  p<\frac{4}{9}  \hbox{ or }  \theta_1<c_1\le0.
\end{array}
%\end{aligned}
\end{equation}
}
\noindent
The special cases $p=\frac{4}{9}, n_1,n_2$ and $c=\theta_1$ need to be considered separately.
Also, we omit  $\exists c_1$ in the above formulations.

In what below, we will give detailed analysis of the above six cases
which leads to the results in Table 1. We first have the formulas:

{%\scriptsize
\footnotesize
\begin{eqnarray}
\mu_{1}- \mu_{3}&=&(p-1)^2c_{1}+2n,\label{eq-m13}\\
\mu_{1}-\mu_{4} &=& \frac{9p((p-1)^2c_1+2n)}{9p-4},\label{eq-m14}\\
\mu_{2}- \mu_{3}&=&\frac{2(p-2)^2((1/2)(p-1)^2c_{1}+n)}{(4p^2-9p+4)},\label{eq-m23}\\
\mu_{2}-\mu_{4} &= &\frac{p(3p-4)^2((p-1)^2c_{1}+2n)}{(4p^2-9p+4)(9p-4)},\label{eq-m24}\\
\mu_{4}- \theta_7 &=& \frac{-(p-1)^2(4n+9p-4)c_{1}+2n^2p(9p-13)}{n(9p-4)},\label{eq-m45}\\
\mu_{3}-\theta_7 &=& \frac{(p-1)(2n^2-pc_{1}+c_{1})}{n},\label{eq-m35}\\
\theta_2-\theta_1 &=& \frac{18n(np-n+1)(p-\frac{4}{9})}{(p-1)^2(4n+9p-4)},\label{eq-c1211}\\
\theta_1-\eta_1 &=& -\frac{2n(np-n+1)}{(p-1)^2},\label{eq-c1114}\\
\eta_1-\theta_2 &=& \frac{8n^2(np-n+1)}{(p-1)^2(4n+9p-4)}.\label{eq-c1412}
\end{eqnarray}
}

We also have the following formulas which will be used to eliminate $c_1$ in the proof.
\begin{equation}
\label{eq-ec1}
\begin{array}{ll}
\mu\le\mu_{3}\Leftrightarrow c_1 \le \theta_4,
&\mu\ge\mu_{4}\Leftrightarrow c_1\ge\theta_5, \hbox{ if } p>\frac{4}{9},\\ \mu\ge\mu_{4}\Leftrightarrow c_1\le\theta_5, \hbox{ if } p<\frac{4}{9},
&\mu\le\mu_{4}\Leftrightarrow c_1\le\theta_5, \hbox{ if } p>\frac{4}{9}, \\ \mu\le\mu_{4}\Leftrightarrow c_1\ge\theta_5, \hbox{ if } p<\frac{4}{9},
&\mu\le\mu_{6}\Leftrightarrow c_1\le\theta_6,\\
\mu\ge\mu_{7}\Leftrightarrow c_{1}\ge\theta_{12},
&\mu\le\mu_{9}\Leftrightarrow c_{1}\le\theta_{13}, \\
\mu\ge\mu_{10}\Leftrightarrow c_{1}\ge\theta_{14},
&\mu\le\mu_{13}\Leftrightarrow c_{1}\le\theta_{15}.\\
\end{array}
\end{equation}

We divide the proof  into several cases, first according to the
values of $c_1$ and then according to the values of $n$.

{\bf Case 1}:  $\theta_1< c_{1}\le0$.
From \eqref{eq-m13}, we have $\mu_{1}> \mu_{3}$ in this case and from \eqref{eq-phi1}, $\Phi(n,p,\mu)$ simplifies to three cases:
\begin{equation}
\label{eq-phic1}
\begin{array}{l}
\Phi(n,p,\mu) \Leftrightarrow 
\max(\mu_{4},\theta_7)\le\mu\le \min(\mu_{2},\mu_{3}), \hbox{ if }  p\in(\frac{4}{9},n_1) \hbox{ or } p> n_2;\\
\Phi(n,p,\mu) \Leftrightarrow \max(\mu_{2},\mu_{4},\theta_7)\le\mu\le \mu_{3}, \hbox{ if } p\in(n_1,n_2);\\
\Phi(n,p,\mu) \Leftrightarrow \theta_7\le\mu\le\min(\mu_{2},\mu_{3}) 
\hbox{ or }
\max(\mu_{4},\theta_7)\le \mu\le\min(\mu_{1},\mu_{2}), \hbox{ if } p<\frac{4}{9}.
\end{array}
\end{equation}
According to the vales of $p$, we consider seven cases below.

{\bf Case 1.1:} $\theta_1< c_{1}\le0$ and $p>n_2$.
In this case, from \eqref{eq-m23} and \eqref{eq-m45},
we have $\mu_{2}\ge \mu_{3}$ and $\mu_{4}> \theta_7$.
From \eqref{eq-phic1}, we have $\Phi(n,p,\mu) \Leftrightarrow \mu_{4}\le \mu\le \mu_{3}$.

We now eliminate $c_1$ from  $\Phi(n,p,\mu) \Leftrightarrow \exists c_1(p>n_2\wedge \theta_1< c_{1}\le0\wedge \mu_{4}\le \mu\le \mu_{3})$.
By \eqref{eq-ec1}, $\mu_{4}\le \mu\le \mu_{3}$ is equivalent to
$\theta_5\le c_1\le\theta_4$. $\exists c_1 (\theta_5\le c_1 \le \theta_4  \wedge \theta_1< c_{1}\le 0$) is equivalent to ($\theta_4> \theta_1\wedge\theta_5\le0$).
Therefore, in this case
$\Phi(n,p,\mu) \Leftrightarrow (p>n_2 \wedge \theta_4> \theta_1\wedge \theta_5\le0)$,
and $\mathbb{T}(1,2)$ is proved.

{\bf Case 1.2:} $\theta_1< c_{1}\le0$ and $p=n_2$. When $p=n_2$, we have $\theta_1=\theta_3,\ s_2=-\frac{1}{1024}(7\sqrt{17}-33)(c_{1}\sqrt{17}+64n+9c_{1})n$. Then $s_2\ge0\Leftrightarrow c_1\ge \theta_3$. Because $\theta_1< c_{1}\le0$ and $p=n_2>\frac{4}{9}$, we have $s_3\ge0\Leftrightarrow\mu_{4}\le\mu\le\mu_{3}$. By \eqref{eq-m45}, we have $\mu_{4}>\theta_7$. When $p=n_2$, we have $\mu_{3}=\mu_{6}$ and $\mu_{4}=\mu_{7}$. Thus $\Phi(n,p,\mu) \Leftrightarrow (\theta_3< c_{1}\le0, \mu_{7}\le \mu\le \mu_{6})$.

We now eliminate $c_1$ from $\Phi(n,p,\mu) \Leftrightarrow \exists c_1(p=n_2\wedge \theta_3< c_{1}\le0\wedge \mu_{7}\le \mu\le \mu_{6})$. By \eqref{eq-ec1}, $\mu_{7}\le \mu\le \mu_{6}$ is equivalent to $\theta_{12}\le c_{1}\le\theta_6$. $\exists c_1 (\theta_{12}\le c_{1}\le\theta_6 \wedge \theta_3< c_{1}\le0)$ is equivalent to $\theta_6> \theta_3$ and $\theta_{12}\le0$. Therefore, in this case
 $\Phi(n,p,\mu) \Leftrightarrow (p=n_2 \wedge \theta_6> \theta_3 \wedge \theta_{12}\le0)$, and $\mathbb{T}(2,2)$ is proved.

{\bf Case 1.3:} $\theta_1< c_{1}\le0$ and $p\in(n_1,n_2)$, $p\neq1$.
Due to \eqref{eq-m45}, this case is divided into two sub-cases.

{\bf Case 1.3.1:} $\theta_1< c_{1}\le0$ and $p\in(\frac{13}{9},n_2)$. By \eqref{eq-m24} and \eqref{eq-m45}, we have $\mu_{4}>\mu_{2}$ and $\mu_{4}>\theta_7$. 
From \eqref{eq-phic1}, we have $\Phi(n,p,\mu) \Leftrightarrow \mu_{4}\le \mu\le \mu_{3}$.

We now eliminate $c_1$ from $\Phi(n,p,\mu) \Leftrightarrow \exists c_1 (p\in(\frac{13}{9},n_2)\wedge \theta_1< c_{1}\le0\wedge \mu_{4}\le \mu\le \mu_{3})$. Similar to  Case 1.1, we have
$\Phi(n,p,\mu) \Leftrightarrow
(p\in(\frac{13}{9},n_2) \wedge \theta_4> \theta_1\wedge \theta_5\le0)$,
$\mathbb{T}(3,2)$ is proved.

{\bf Case 1.3.2:} $\theta_1< c_{1}\le0$ and $p\in(n_1,\frac{13}{9}]$, $p\neq1$. By \eqref{eq-m24}, \eqref{eq-m35}, \eqref{eq-m45} and \eqref{eq-c1211}, we have $\mu_{4}\ge\mu_{2}$, ($\mu_{3}\ge\theta_7\Leftrightarrow c_1\le \eta_1$), ($\mu_{4}\ge\theta_7\Leftrightarrow c_{1}\le \theta_2$) and ($\theta_2> \theta_1\Leftrightarrow \phi_3$). Hence $\Phi(n,p,\mu) \Leftrightarrow \max(\mu_{4},\theta_7)\le \mu\le \mu_{3}$.
This case is further divided into two sub-cases.

{\bf Case 1.3.2.1:} If $c_{1}\le \theta_2$, then $\mu_{4}\ge\mu_{5}$, and $\Phi(n,p,\mu) \Leftrightarrow \mu_{4}\le \mu\le \mu_{3}$. So we need $\theta_1< \theta_2$, which yields $\phi_3$. From \eqref{eq-phic1}, we have $\Phi(n,p,\mu) \Leftrightarrow (\theta_1<c_1\le \theta_2, \phi_3, \mu_{4}\le \mu\le \mu_{3})$.

We now eliminate $c_1$ from $\Phi(n,p,\mu) \Leftrightarrow\exists c_1 (p\in(n_1,\frac{13}{9})\wedge p\neq1\wedge \phi_3\wedge \theta_1<c_1\le \theta_2\wedge \mu_{4}\le \mu\le \mu_{3})$. Like Case 1.1, we have
$\Phi(n,p,\mu) \Leftrightarrow
(p\in(n_1,\frac{13}{9})\wedge p\neq1 \wedge
\phi_3 \wedge \theta_4> \theta_1\wedge \theta_5\le \theta_2)$,
 and $\mathbb{T}(4,2)$ is proved

{\bf Case 1.3.2.2:} If $c_{1}\ge \theta_2$, then $\mu_{4}\le\theta_7$, and $\Phi(n,p,\mu) \Leftrightarrow \theta_7\le \mu\le \mu_{3}$. So we need $\theta_7\le\mu_{3}$, which yields $c_{1}\le \eta_1$. By \eqref{eq-c1114}, we know $\eta_1> \theta_1$ results $\phi_3$, which yields $\theta_1< \theta_2< \eta_1$.
From \eqref{eq-phic1}, we have $\Phi(n,p,\mu) \Leftrightarrow (\theta_2< c_{1}\le \min(0,\eta_1), \phi_3, \theta_7\le \mu\le \mu_{3})$.

We now eliminate $c_1$ from
$\Phi(n,p,\mu) \Leftrightarrow\exists c_1 (p\in(n_1,\frac{13}{9})\wedge
 p\neq1\wedge \phi_3\wedge \theta_2< c_{1}\le \min(0,\eta_1)\wedge \theta_7\le \mu\le \mu_{3})$. $\theta_7\le \mu\le \mu_{3}$ is equivalent to $\theta_7 \le \mu$ and $c_1 \le \theta_4$. $\exists c_1 (c_1 \le \theta_4  \wedge \theta_2< c_{1}\le \min(0,\eta_1)$)
is equivalent to $\theta_4 > \theta_2$. Therefore, in this case
 $\Phi(n,p,\mu) \Leftrightarrow
(p\in(n_1,\frac{13}{9})\wedge p\neq1\wedge
\phi_3 \wedge \theta_7 \le \mu \wedge \theta_4 > \theta_2)$,
and $\mathbb{T}(4,3)$ is proved.

{\bf Case 1.4:} $\theta_1< c_{1}\le0$ and $p=n_1$. When $p=n_1$, we have $\theta_1=\theta_{10},\ \theta_2=\theta_{11},\ \eta_1=\eta_2,\ s_2=-\frac{1}{1024}(33+7\sqrt{17})(c_{1}\sqrt{17}-64n-9c_{1})n$. Then $s_2\ge0\Leftrightarrow c_1\ge \theta_{10}$. Because $\theta_1< c_{1}\le0$ and $p=n_1>\frac{4}{9}$, we have $s_3\ge0\Leftrightarrow\mu_{4}\le\mu\le\mu_{3}$. By \eqref{eq-m45}, we have $\mu_{4}\ge\theta_7\Leftrightarrow c_{1}\le \theta_2$.
This case is divided into two sub-cases.

{\bf Case 1.4.1:} Similar to  Case 1.3.2.1, $\Phi(n,p,\mu) \Leftrightarrow (\theta_1<c_{1}\le \theta_2 \wedge \phi_3\wedge \mu_{4}\le \mu\le \mu_{3})$. When $p=n_1$, we have $\mu_{3}=\mu_{9},\ \mu_{4}=\mu_{10},\ \phi_3\Leftrightarrow n< n_3$.
From \eqref{eq-phic1}, we have $\Phi(n,p,\mu) \Leftrightarrow (\theta_{10}<c_{1}\le \theta_{11}, n< n_3, \mu_{10}\le \mu\le \mu_{9})$.

We now eliminate $c_1$ from
$\Phi(n,p,\mu) \Leftrightarrow \exists c_1(p=n_1\wedge \theta_{10}<c_{1}\le \theta_{11}\wedge n< n_3\wedge \mu_{10}\le \mu\le \mu_{9})$. $\mu_{10}\le \mu\le \mu_{9}$ is equivalent to $\theta_{14}\le c_{1}\le\theta_{13}$. $\exists c_1 (\theta_{14}\le c_{1}\le\theta_{13} \wedge \theta_{10}< c_{1}\le \theta_{11})$ is equivalent to $\theta_{13}> \theta_{10}$ and $\theta_{14}\le \theta_{11}$. Therefore, in this case
 $\Phi(n,p,\mu) \Leftrightarrow (p=n_1\wedge
n< n_3 \wedge \theta_{13}> \theta_{10} \wedge \theta_{14}\le \theta_{11})$, and $\mathbb{T}(5,2)$ is proved.

{\bf Case 1.4.2:} Similar to Case 1.3.2.2, $\Phi(n,p,\mu) \Leftrightarrow (\theta_2< c_{1}\le \min(0,\eta_1)\wedge \phi_3\wedge \mu_{5}\le \mu\le \mu_{3})$.
When $p=n_1$, we have $\theta_7=\theta_{16}$.
From \eqref{eq-phic1}, we have $\Phi(n,p,\mu) \Leftrightarrow (\theta_{11}< c_{1}\le \eta_2, n< n_3, \theta_{16}\le \mu\le \mu_{9})$.

We now eliminate $c_1$ from
 $\Phi(n,p,\mu) \Leftrightarrow \exists c_1(p=n_1\wedge 
 \theta_{11}< c_{1}\le \eta_2\wedge n< n_3\wedge \theta_{16}\le \mu\le \mu_{9})$. $\theta_{16}\le \mu\le \mu_{9}$ is equivalent to $c_{1}\le\theta_{13}$ and $\mu\ge \theta_{16}$. $\exists c_1 (c_{1}\le\theta_{13} \wedge \theta_{11}< c_{1}\le \eta_2)$ is equivalent to $\theta_{13}> \theta_{11}$. Therefore, in this case
 $\Phi(n,p,\mu) \Leftrightarrow
 (p=n_1\wedge
 n< n_3 \wedge \theta_{13}> \theta_{11} \wedge \mu\ge \theta_{16})$, and $\mathbb{T}(5,3)$ is proved.

{\bf Case 1.5:} $\theta_1< c_{1}\le0$ and $p\in(\frac{4}{9},n_1)$. By \eqref{eq-m23} and \eqref{eq-m45}, we have $\mu_{2}>\mu_{3}$ and ($\mu_{4}\ge\theta_7\Leftrightarrow c_{1}\le \theta_2$). Hence $\Phi(n,p,\mu) \Leftrightarrow \max(\mu_{4},\theta_7)\le \mu\le \mu_{3}$.
This case is divided into two sub-cases.

{\bf Case 1.5.1:} Similar to  Case 1.3.2.1, we have $\Phi(n,p,\mu) \Leftrightarrow (\theta_1<c_{1}\le \theta_2\wedge \phi_3\wedge \mu_{4}\le \mu\le \mu_{3})$.

We now eliminate $c_1$ from $\Phi(n,p,\mu) \Leftrightarrow \exists c_1(p\in(\frac{4}{9},n_1)\wedge \phi_3\wedge \theta_1<c_{1}\le \theta_2\wedge \mu_{4}\le \mu\le \mu_{3})$. Like  Case 1.3.2.1,
we have $\Phi(n,p,\mu) \Leftrightarrow
(p\in(\frac{4}{9},n_1) \wedge
\phi_3 \wedge \theta_4> \theta_1\wedge \theta_5\le \theta_2)$,
 and $\mathbb{T}(6,2)$ is proved

{\bf Case 1.5.2:} Similar to  Case 1.3.2.2, we have $\Phi(n,p,\mu) \Leftrightarrow (\theta_2< c_{1}\le \eta_1\wedge \phi_3\wedge \mu_{5}\le \mu\le \mu_{3})$.

We now eliminate $c_1$ from $\Phi(n,p,\mu) \Leftrightarrow\exists c_1 (p\in(\frac{4}{9},n_1)\wedge \phi_3\wedge \theta_2< c_{1}\le \eta_1\wedge \theta_7\le \mu\le \mu_{3})$.
Similar to  Case 1.3.2.2, we have
 $\Phi(n,p,\mu) \Leftrightarrow
 (p\in(\frac{4}{9},n_1)\wedge
 \phi_3 \wedge \theta_7 \le \mu \wedge \theta_4 > \theta_2)$,
 and $\mathbb{T}(6,3)$ is proved.

{\bf Case 1.6:} $\theta_1< c_{1}\le0$ and $p=\frac{4}{9}$. When $p=\frac{4}{9}$, we have $\theta_1=\theta_9,\ \eta_1=\eta_3,\ \theta_7=\theta_8,\ s_3=-\frac{4n}{6561}(162n+25c_{1})(45n^2+81nu+25c_{1})$. Then $s_3\ge0\Leftrightarrow \mu\le\mu_{13}$ if $c_{1}\ge \theta_9$. By \eqref{eq-m35}, we know $\mu_{3}\ge\theta_7\Leftrightarrow c_1\le \eta_1$. And $\theta_9< \eta_3\Leftrightarrow n=1$. 
From \eqref{eq-phic1}, we have $\Phi(n,p,\mu) \Leftrightarrow (\theta_9<c_{1}\le \eta_3, n=1, \theta_8\le \mu\le \mu_{13})$.

We now eliminate $c_1$ from
 $\Phi(n,p,\mu) \Leftrightarrow\exists c_1 
 (p=\frac{4}{9}\wedge \theta_9<c_{1}\le \eta_3\wedge n=1\wedge \theta_8\le \mu\le \mu_{13})$. $\theta_8\le \mu\le \mu_{13}$ is equivalent to $c_{1}\le\theta_{15}$ and $\mu\ge\theta_8$. $\exists c_1 (c_{1}\le\theta_{15} \wedge \theta_9<c_{1}\le \eta_3)$ is equivalent to $\theta_{15}>\theta_9$. Therefore, in this case
 $\Phi(n,p,\mu) \Leftrightarrow
 (p=\frac{4}{9} \wedge
 n=1 \wedge \theta_{15}> \theta_9 \wedge \mu\ge \theta_8)$,
 and $\mathbb{T}(7,2)$ is proved

{\bf Case 1.7:} $\theta_1< c_{1}\le0$ and $0<p<\frac{4}{9}$.
This case is divided into two sub-cases.

{\bf Case 1.7.1:} If we select $\theta_7\le\mu\le\min(\mu_{2},\mu_{3})$, by \eqref{eq-m23}, we have $\mu_{2}>\mu_{3}$. Thus $\Phi(n,p,\mu) \Leftrightarrow \theta_7\le \mu\le \mu_{3}$. So we need $\theta_7\le\mu_{3}$, which yields $c_{1}\le \eta_1$. By \eqref{eq-c1114}, we know $\eta_1> \theta_1$ results $\phi_3$, which yields $n=1$ with $0<p<\frac{4}{9}$. 
From \eqref{eq-phic1}, we have $\Phi(n,p,\mu) \Leftrightarrow (\theta_1< c_{1}\le \eta_1, n=1, \theta_7\le \mu\le \mu_{3})$.

We now eliminate $c_1$ from
$\Phi(n,p,\mu) \Leftrightarrow \exists c_1(p\in(0,\frac{4}{9})\wedge n=1\wedge \theta_1< c_{1}\le \eta_1\wedge \theta_7\le \mu\le \mu_{3})$. Similar to Case 1.3.2.2, we have $\Phi(n,p,\mu) \Leftrightarrow
(p\in(0,\frac{4}{9})\wedge
n=1 \wedge \theta_7 \le \mu \wedge \theta_4 > \theta_1)$,
and $\mathbb{T}(8,2)$ is proved.

{\bf Case 1.7.2:} If we select $\max(\mu_{4},\theta_7)\le \mu\le\min(\mu_{1},\mu_{2})$, by \eqref{eq-m14}, we have $\mu_{1}<\mu_{4}$, which yields contradiction.

{\bf Case 2:}  $c_{1}<\theta_1$.
From \eqref{eq-m13}, we have $\mu_{1}< \mu_{3}$ in this case and from \eqref{eq-phi0}, $\Phi(n,p,\mu)$ simplifies to one case:
$$\Phi(n,p,\mu) \Leftrightarrow\exists c_1
(\theta_7\le\mu\le\min(\mu_{1},\mu_{2},\mu_{4}), 
\hbox{ if } 0<p<\frac{4}{9} \hbox{ and }c_1< \theta_1).$$
Since $p$ satisfies $0<p<\frac{4}{9}$, we need only consider the following cases.

{\bf Case 2.1:} $c_1<\theta_1$ and $0<p<\frac{4}{9}$. By \eqref{eq-m14}, \eqref{eq-m24} and \eqref{eq-m45}, we have $\mu_{1}>\mu_{4}$, $\mu_{2}>\mu_{4}$ and ($\mu_{4}\ge\theta_7\Leftrightarrow c_1\ge \theta_2$). Then we need $\theta_2< \theta_1$, which yields $\phi_3$ by \eqref{eq-c1211}. Because $\phi_3$ means $n=1$ with $0<p<\frac{4}{9}$, we have $\Phi(n,p,\mu) \Leftrightarrow (\theta_2\le c_{1}< \theta_1, n=1, \theta_7\le \mu\le \mu_{4})$.

We now eliminate $c_1$ from
 $\Phi(n,p,\mu) \Leftrightarrow \exists c_1(p\in(0,\frac{4}{9})\wedge n=1\wedge \theta_2\le c_{1}< \theta_1\wedge \theta_7\le \mu\le \mu_{4})$. $\theta_7\le \mu\le \mu_{4}$ is equivalent to $c_{1}\ge\theta_5$ and $\mu\ge\theta_7$. $\exists c_1 (c_{1}\ge\theta_5 \wedge \theta_2\le c_{1}< \theta_1)$ is equivalent to $\theta_5<\theta_1$. Therefore, in this case
 $\Phi(n,p,\mu) \Leftrightarrow
(p\in(0,\frac{4}{9}) \wedge
n=1 \wedge \theta_5< \theta_1 \wedge \mu\ge \theta_7)$,
 and $\mathbb{T}(8,3)$ is proved

{\bf Case 2.2:} $c_{1}<\theta_1$ and $p=n_2$. In  Case 1.2, we know $\theta_1=\theta_3$ with $p=n_2$, and $s_2\ge0\Leftrightarrow c_1\ge \theta_3$, which yields contradiction.

{\bf Case 2.3:} $c_{1}<\theta_1$ and $p=n_1$. In  Case 1.4, we know $\theta_1=\theta_{10}$ with $p=n_1$, and $s_2\ge0\Leftrightarrow c_1\ge \theta_{10}$, which yields contradiction.

{\bf Case 2.4:} $c_{1}<\theta_1$ and $p=\frac{4}{9}$. We have $\theta_1=\theta_9,\ \mu_{2}=\mu_{12},\ \mu_{3}=\mu_{13}$ based on $p=\frac{4}{9}$. Then we have ($s_2\ge0\Leftrightarrow \mu\le\mu_{12}$) and ($s_3\ge0\Leftrightarrow \mu\ge\mu_{13}$ if $c_{1}\le \theta_9$). So we need $\mu_{12}\ge \mu_{13}$. By \eqref{eq-m23}, we have $\mu_{12}<\mu_{13}$, which yields contradiction.

{\bf Case 3:} $c_{1}=\theta_1$. When $c_{1}=\theta_1$, we have $s_1=n(np-n-\mu+2)$, $s_2=n(4p^2-9p+4)(np-n-\mu+2)$ and $s_3=-n^2(9p-4)(np-n-\mu+2)^2$. Thus, $s_1\ge0\Leftrightarrow \mu\le 2+n(p-1)$ and $s_3\ge0\Leftrightarrow (p\le\frac{4}{9}\ {\rm or}\ \mu=2+n(p-1))$.
This case is divided into two sub-cases.

{\bf Case 3.1:} If $\mu=2+n(p-1)$, then $s_1=s_2=s_3=0$. And $p\ge1-\frac{\mu}{n}\Leftrightarrow \phi_1$. Thus $\Phi(n,p,\mu) \Leftrightarrow (c_{1}= \theta_1\wedge\phi_1\wedge\phi_2)$,
and $\mathbb{T}(i,4),\ i=1,\ldots,6$ are proved.

{\bf Case 3.2:} If $p\le\frac{4}{9}$, then $s_2\ge0\Leftrightarrow \mu\le 2+n(p-1)$. Then we need $2+n(p-1)\ge\mu_{5}$, which yields $\phi_1$. And $\phi_1$ implies $n=1$ with $p\le\frac{4}{9}$. Thus $\Phi(n,p,\mu) \Leftrightarrow (c_{1}= \theta_1\wedge n=1\wedge-p\le\mu-1\le p)$,
and
$\mathbb{T}(7,4)$, $\mathbb{T}(8,4)$ are proved.
%\end{proof}

\section{Conclusion}
\label{sec-conc}

This paper is an extension  of the work \cite{GYG2020,GYG2020M}
to the case where the entropy power involves parameters.
The basic idea is to prove entropy power inequalities in a systematical way.
Precisely, the concavity of R\'{e}nyi entropy power is considered,
where the probability density $u_t$ solve
the nonlinear heat equation with two parameters $p$ and $\mu$.
Our procedure reduces the proof of the CREP to check the
semi-positiveness of a quadratic form \eqref{eq-A1} whose coefficients are polynomials in the parameters $n,p,\mu$.
%We should mention that  \cite{GYG2020,GYG2020M}
%
In principle, a necessary and sufficient condition on parameters $n,p,\mu$ for
this can be computed with the quantifier elimination~\cite{QE}.
%In this paper, the problem is special and an explicit proof is given.

Based on the above method, we give a sufficient condition $\Phi(n,p,\mu)$ for CREP, which extends the parameter's range of the CREP given by Savar\'{e}-Toscani~\cite{Savare2014}.
In fact, our results give the necessary and sufficient condition for CREP under certain conditions.
But in the general case,  Theorem \ref{th-m1} only gives a sufficient condition
for the following reasons:
1. we use inequalities \eqref{T2} in step \eqref{T33};
2. there might exist more  constraints;
3. {\bf Problem II} may not equivalent to {\bf Problem III}.

\section*{Acknowledgments}

This work is partially supported by NSFC 11688101 and NKRDP 2018YFA0704705, Beijing Natural Science Foundation (No. Z190004),
and China Postdoctoral Science Foundation (No. 2019TQ0343, 2019M660830).

%\appendix

\section*{Appendix. Constraints in \eqref{eq-2cons}}
In this appendix, we give the constraints in \eqref{eq-2cons}, where
$u_{h_1,h_2}= \frac{\partial^{h_1+h_2} u}{\partial^{h_1} x_{a,t}\partial^{h_2} x_{b,t}}$.

{\parskip=3pt
%\tiny
%\begin{equation*}
%\begin{array}{ll}
$R_{1,a,b}=3pu_{1,0}^4+3uu_{2,0}u_{1,0}^2-5u_{1,0}^4,$

$R_{2,a,b}=3pu_{0,1}^4+3uu_{0,2}u_{0,1}^2-5u_{0,1}^4,$

$R_{3,a,b}=3pu^2u_{1,0}u_{3,0}+u^3u_{4,0}-3u^2u_{1,0}u_{3,0},$

$R_{4,a,b}=3pu^2u_{0,1}u_{3,0}+u^3u_{3,1}-3u^2u_{0,1}u_{3,0},$

$R_{5,a,b}=3pu^2u_{1,0}u_{2,1}+u^3u_{3,1}-3u^2u_{1,0}u_{2,1},$

$R_{6,a,b}=3pu^2u_{0,1}u_{2,1}+u^3u_{2,2}-3u^2u_{0,1}u_{2,1},$

$R_{7,a,b}=3pu^2u_{1,0}u_{1,2}+u^3u_{2,2}-3u^2u_{1,0}u_{1,2},$

$R_{8,a,b}=3pu^2u_{0,1}u_{1,2}+u^3u_{1,3}-3u^2u_{0,1}u_{1,2},$

$R_{9,a,b}=3pu^2u_{1,0}u_{0,3}+u^3u_{1,3}-3u^2u_{1,0}u_{0,3},$

$R_{10,a,b}=3pu^2u_{0,1}u_{0,3}+u^3u_{0,4}-3u^2u_{0,1}u_{0,3},$

$R_{11,a,b}=3pu_{0,1}u_{1,0}^3+3uu_{1,1}u_{1,0}^2-5u_{1,0}^3u_{0,1},$

$R_{12,a,b}=3pu_{0,1}^3u_{1,0}+3uu_{1,1}u_{0,1}^2-5u_{0,1}^3u_{1,0},$

$R_{13,a,b}=3pu_{0,1}^2u_{1,0}^2+2uu_{0,1}u_{1,0}u_{1,1}+uu_{1,0}^2u_{0,2}-5u_{1,0}^2u_{0,1}^2,$

$R_{14,a,b}=3puu_{1,0}^2u_{2,0}+u^2u_{1,0}u_{3,0}+u^2u_{2,0}^2-4uu_{2,0}u_{1,0}^2,$

$R_{15,a,b}=3puu_{0,1}u_{1,0}u_{2,0}+u^2u_{0,1}u_{3,0}+u^2u_{1,1}u_{2,0}-4uu_{2,0}u_{1,0}u_{0,1},$

$R_{16,a,b}=3puu_{1,0}^2u_{1,1}+u^2u_{1,0}u_{2,1}+u^2u_{1,1}u_{2,0}-4uu_{1,1}u_{1,0}^2,$

$R_{17,a,b}=3puu_{0,1}u_{1,0}u_{2,0}+u^2u_{1,0}u_{2,1}+u^2u_{1,1}u_{2,0}-4uu_{2,0}u_{1,0}u_{0,1},$

$R_{18,a,b}=3puu_{0,1}^2u_{2,0}+u^2u_{0,1}u_{2,1}+u^2u_{0,2}u_{2,0}-4uu_{0,1}^2u_{2,0},$

$R_{19,a,b}=3puu_{0,1}u_{1,0}u_{1,1}+u^2u_{0,1}u_{2,1}+u^2u_{1,1}^2-4uu_{0,1}u_{1,0}u_{1,1},$

$R_{20,a,b}=3puu_{1,0}^2u_{0,2}+u^2u_{1,0}u_{1,2}+u^2u_{0,2}u_{2,0}-4uu_{1,0}^2u_{0,2},$

$R_{21,a,b}=3puu_{0,1}u_{1,0}u_{1,1}+u^2u_{1,0}u_{1,2}+u^2u_{1,1}^2-4uu_{0,1}u_{1,0}u_{1,1},$

$R_{22,a,b}=3puu_{0,1}^2u_{1,1}+u^2u_{0,1}u_{1,2}+u^2u_{0,2}u_{1,1}-4uu_{1,1}u_{0,1}^2,$

$R_{23,a,b}=3puu_{0,1}u_{1,0}u_{0,2}+u^2u_{0,1}u_{1,2}+u^2u_{0,2}u_{1,1}-4uu_{0,2}u_{1,0}u_{0,1},$

$R_{24,a,b}=3puu_{0,1}u_{1,0}u_{0,2}+u^2u_{1,0}u_{0,3}+u^2u_{0,2}u_{1,1}-4uu_{0,2}u_{1,0}u_{0,1},$

$R_{25,a,b}=3puu_{0,1}^2u_{0,2}+u^2u_{0,1}u_{0,3}+u^2u_{0,2}^2-4uu_{0,2}u_{0,1}^2,$

$R_{26,a,b}=3pu_{0,1}u_{1,0}^3+2uu_{2,0}u_{1,0}u_{0,1}+uu_{1,1}u_{1,0}^2-5u_{1,0}^3u_{0,1},$

$R_{27,a,b}=3pu_{0,1}^2u_{1,0}^2+uu_{0,1}^2u_{2,0}+2uu_{0,1}u_{1,0}u_{1,1}-5u_{1,0}^2u_{0,1}^2,$

$R_{28,a,b}=3pu_{0,1}^3u_{1,0}+uu_{1,1}u_{0,1}^2+2uu_{0,2}u_{1,0}u_{0,1}-5u_{0,1}^3u_{1,0}.$
%\end{array}
%\end{equation*}
}

\end{document}